\documentclass[journal]{IEEEtran}
\IEEEoverridecommandlockouts
\usepackage{cite}
\usepackage{amsmath,amssymb,amsfonts,amsthm,epsfig,nccmath}
\usepackage{algorithmic}
\usepackage{graphicx,float,color, soul}
\usepackage{textcomp, hyperref,subcaption}
\usepackage{bbm}
\newtheorem{lemma}{Lemma}
\usepackage{tikz, verbatim}
\usepackage{relsize}
\usepackage{svg}
\usepackage{soul}
\usepackage{setspace}
\usepackage{amsmath}               

\newcommand{\dft}[1]{\mathfrak{F}({#1})}
\newcommand{\idft}[1]{\mathfrak{F}^{-1}({#1})}
\def \S{{ \mathbb{S} }} 

\def \X{{ \tilde{X} }}
\def \Y{{ Y  }}

\def \P{{ \mathcal{P} }}
\def \Q{{ \mathcal{Q} }}
\def \T{{ \mathcal{T} }}
\def \Z{{ \mathcal{Z} }}
\def \F{{ \mathcal{F} }}
\def \U{{ \mathcal{U} }}
\def \R{{ \mathbb{R} }}
\usepackage{array}
\newcolumntype{P}[1]{>{\centering\arraybackslash}p{#1}}

\def \l{{ l }}
\def \g{{ \tilde{g} }}
\def \G{{ \tilde{G} }}

\def \u{{ u }}
\usepackage{float}

\newtheorem{assumption}{Assumption}

\newtheorem{theorem}{Theorem}[section]
\usepackage[linesnumbered,ruled,vlined]{algorithm2e}

\SetCommentSty{mycommfont}
\SetKwComment{Comment}{$\triangleleft$\ }{}

\newcommand{\revv}[1]{{\color{black}#1}}
\newcommand{\edit}[1]{{\color{black}#1}}

\renewcommand{\ul}[1]{\unskip}
\renewcommand{\st}[1]{\unskip}

\title{Signal Reconstruction from Quantized Noisy Samples of \ul{Phase of} the
Discrete Fourier Transform}

\author{\IEEEauthorblockN{Mohak Goyal\thanks{Most of this work was done when Mohak Goyal was with the Department of Electrical Engineering at IIT Bombay. He is now with the Department of Management Science \& Engineering at Stanford University.} and Animesh Kumar}

\IEEEauthorblockA{Department of Electrical Engineering, Indian Institute of
Technology Bombay, India\\
goyalmohak2@gmail.com, animekum@outlook.com }}

\begin{document}

\maketitle

\begin{abstract}
In this paper, we present two variations of an algorithm for signal
reconstruction from one-bit or two-bit noisy observations of the \ul{phase of
the} discrete Fourier transform (DFT). The one-bit observations of the DFT
\ul{phase} correspond to the sign of its real part, whereas, the two-bit
observations of the DFT \ul{phase} correspond to the signs of both the real and
imaginary parts of the DFT. We focus on images for analysis and simulations,
thus using the sign of the 2D-DFT. This choice of the class of signals is
inspired by previous works on this problem. For our algorithm, we show that the
expected mean squared error (MSE) in signal reconstruction is asymptotically
proportional to the inverse of the sampling rate. The samples are affected by
additive zero-mean noise of known distribution. We solve this signal estimation
problem by designing an algorithm that uses contraction mapping, based on the
Banach fixed point theorem. Numerical tests with four benchmark images are
provided to show the effectiveness of our algorithm. Various metrics for image
reconstruction quality assessment such as PSNR, SSIM, ESSIM, and MS-SSIM are
employed. On all four benchmark images, our algorithm outperforms the
state-of-the-art in all of these metrics by a significant margin.
\end{abstract}

\section{Introduction}
\label{sec:intro}
Signal reconstruction from partial information of its Fourier transform (FT) has
been of interest since a long time, both for its practical applications and for
its fundamental understanding~\cite{oppenheim1983signal}. Two well-known types
of problems in this class are the phase retrieval and magnitude retrieval
problems~\cite{oppenheim1983signal,fienup2013phase}. The phase retrieval problem
is the reconstruction or estimation of a signal from its FT magnitude
information. It has applications in areas such as electron
microscopy~\cite{saxton2013computer} and X-ray
crystallography~\cite{ramachandran1970fourier}. Recent developments in phase
retrieval include its formulation as a semi-definite program with robustness
guarantees by Cand\`es, Strohmer, and Voroninski~\cite{candes2013phaselift}.
Huang, Eldar, and Sidiropoulos~\cite{huang2016phase} gave a polynomial time
algorithm with uniqueness and optimality guarantees for the phase retrieval
problem. Phase retrieval from one-bit measurements of the magnitude has been
studied in~\cite{kishore2020wirtinger} and~\cite{mukherjee2018phase}, among
others.

The magnitude retrieval problem is the reconstruction of a signal from partial
information of its FT phase. Its applications are in situations where the signal
is distorted by a zero-phase blurring or point-spread function. In such cases
the magnitude information is lost but the phase is retained. Li and Kurkjian
~\cite{li1983arrival} showed that magnitude retrieval can be used to solve some
problems in arrival-time estimation. An interesting version of the magnitude
retrieval problem is the case where only coarsely quantized readings of the FT
phase of the signal are available. The signal reconstruction problem with
one-bit readings of FT phase has also been studied. Curtis, Oppenheim, and Lim
showed that most two-dimensional signals can be reconstructed to within a scale
factor from only one bit of 2D-DFT phase~\cite{curtis1985signal}. They presented
an iterative algorithm which consists of projection onto the support region in
the spatial domain and enforcement of phase information in the frequency
domain. Tang, Yuan, and Wang presented an improvement to this algorithm with a
specified histogram constraint~\cite{tang1990image}. 

Lyuboshenko and Akhmetshin studied signal reconstruction with noisy FT
phase~\cite{lyuboshenko1999stable}. They proposed global and local
regularization based reconstruction algorithms. Working on signal reconstruction
with noisy FT phase, Thomas and Hayes proposed algorithms which incorporate side
information such as a bound on the noise and the 2D-DFT
magnitudes~\cite{thomas1984procedures}. However, these works considered
full-precision measurements only. Unlike these previous works, in this paper we
consider the \ul{magnitude retrieval} \revv{signal reconstruction} problem with
two deficiencies in the available information of the FT \ul{phase}: (1) the
recordings are corrupted with additive zero-mean noise, and (2) the recordings
are quantized with precision of only one or two bits \ul{per sample of the FT
phase}.

Signal reconstruction from coarsely quantized samples is well known in classical
signal processing~\cite{gray1987oversampled, masry, thong2002nonlinear,
daubechies2003approximating, kumar2013estimation,cvetkovic2000single} and is
particularly appealing in hardware implementations. The quantizer to one-bit is
a comparator to zero and is quite fast, thus enabling high sampling rates.
Dealing with one-bit recordings in the time (or spatial) domain, Kumar and
Prabhakaran obtained a mean squared error of $O(1/K)$ for (classically)
bandlimited signals with $K$ being the oversampling factor with respect to the
Nyquist rate~\cite{kumar2013estimation}. Cvetkovi\`c, Daubechies, and Logan
considered the case of irregular sampling and used a deterministic dither,
obtaining a $O(1/K)$ pointwise error~\cite{cvetkovic2000single}. Khobahi et. al.
have employed deep neural networks for one-bit signal
recovery~\cite{khobahi2019deep}. The application of one-bit samples for channel
estimation has been explored in~\cite{bender2020spectral}
and~\cite{shao2019channel}, among others. Another application of one-bit
quantization has been shown in graph signal processing for bandlimited graph
signals by Goyal and Kumar~\cite{goyal2018estimation}.

Compressed sensing with one-bit samples was introduced by Boufounos and
Baraniuk~\cite{boufounos20081}. They gave a convex relaxation of the problem,
employing a one-sided quadratic penalty. Zymnis, Boyd, and Cand\`es gave two
algorithms for compressed sensing with one-bit samples, based on $l_1$
regularised least squares and $l_1$ regularised maximum
likelihood~\cite{zymnis2010compressed}. Xu and Jacques proposed an algorithm
that utilizes a random dither~\cite{xu2018quantized} for one-bit compressed
sensing. Jacques et. al., gave the binary iterative hard thresholding (BIHT)
algorithm which is also robust to noise~\cite{jacques2013robust}. \revv{For
noiseless signals, Friedlander et al.  proved that a variant of BIHT achieves
the optimal $O(1/K)$ error decay rate with high
probability~\cite{friedlander2020nbiht}.} We use the BIHT algorithm as a
baseline while evaluating the performance of our algorithm. 

\revv{Boufounos introduced angle-preserving quantized phase
embeddings~\cite{boufounos2013angle}. They consider a real valued signal, for
which phase measurements were obtained through a complex linear transform. They
showed that these embeddings generalize the binary epsilon stable embeddings in
the same sense that the phase of complex numbers generalizes the sign of real
numbers. Boufounos introduced complex compressive sensing where measurements of
a sparse signal are obtained via a complex, fat sensing
matrix~\cite{boufounos2013sparse}. They proved that with complex Gaussian
random sensing matrices, one can estimate the direction of such a signal from
the phase of the compressive measurements. Jacques and Feuillen extended this
idea to any signals belonging to a symmetric, low-complexity conic set of
reduced dimensionality, including the set of sparse signals or the set of
low-rank matrices~\cite{jacques2020importance}. These works consider only
noiseless signals.}

In this work, we consider the \ul{magnitude retrieval} problem of signal
estimation from noisy quantized readings of the DFT \ul{phase}. With regard to
quantization, we restrict our attention to one-bit or two-bit precision only.
The one-bit observations of the DFT correspond to the sign of its real part,
whereas, the two-bit observations of the DFT correspond to the signs of both the
real and imaginary parts of the DFT. Similar to Kumar and Prabhakaran
in~\cite{kumar2013estimation}, we provide a Banach contraction mapping based
algorithm for estimation and denoising. Unlike their work, the algorithm we
present can also be used to estimate non-bandlimited signals. \ul{In the case of
one-bit quantization of the DFT, we have the information corresponding to the
sign of the real part of the DFT. In this case, we need sufficient zero-padding
of the signal in the time domain (or spatial domain for images) to ensure that
the real part of the DFT (or 2D-DFT) phase contains complete information of the
signal, as discussed in}

\revv{Our algorithm achieves the optimal $O(1/K)$ error decay rate. This rate
has also been observed by other signal reconstruction approaches, for example in
linear regression~\cite{seber2012linear}, with quantized
samples~\cite{goyal1998quantized, cvetkovic2000single, kumar2013estimation}, and
in compressed sensing \cite{ai2014one, plan2013one}. While there exists
literature on compressed sensing with one-bit samples, (see
e.g.,~\cite{boufounos20081, xu2018quantized, zymnis2010compressed,
jacques2013robust, friedlander2020nbiht, ai2014one, plan2013one}) there has not
been much work on denoising of `lowpass' signals from noise affected and
single-bit quantized samples. Thus, there is significant difference in the
signal model (lowpass versus sparse), presence or absence of noise, and
guarantees (analytical results versus recovery algorithms with simulations)
between this work and the compressed sensing literature.  We sample only the
sign of the real and imaginary parts of the noisy DFT, and use the distribution
of the noise to perform signal reconstruction with $O(1/K)$ error decay rate.} 

There are two variations of the algorithm we present. Algorithms~\ref{algo:1}
and~\ref{algo:2} solve the signal reconstruction problem with one-bit and
two-bit recordings respectively. \ul{For the case of two-bit quantization, we
obtain the sign of both the real and imaginary parts of the DFT.} We use
Banach's contraction mapping theorem to prove that our algorithm converges to a
unique point. We also provide a proof of the decay of the expected mean squared
error (MSE) for our algorithm with increase in sampling rate. It is inversely
proportional to the sampling rate. In this paper we consider grayscale images
\ul{only}. We demonstrate the results on the IEEE logo image, the Lena image,
the cameraman image, and the peppers image. We also show comparisons with two
current \ul{state-of-the-art} algorithms.

\textbf{Paper outline:} The problem setup is described in
Section~\ref{sec:model}. The algorithms are given in Section~\ref{sec:algo}. In
Section~\ref{sec:proofs}, bounds on the estimation error are given with proofs.
Simulations and their results are explained in Section~\ref{sec:sim}.
Conclusions are drawn in Section~\ref{sec:conclusion}.

\textbf{Notation:} We use standard notation: the set of real numbers is denoted
by $\R$, the complex numbers by $\mathbb{C}$, the 2D-discrete Fourier transform
(2D-DFT) of the real image $g[n_1,n_2]$ is denoted by $\g[k_1,k_2]$.  The
indicator function, denoted by $\mathbbm{1}(x > 0)$, takes value $1$ if $x > 0$
and $0$ otherwise.  The Frobenius norm and the \ul{infinity} \revv{max} norm of
matrix $\mathcal{M}$ are given by $\|\mathcal{M}\|_F$ and
$\|\mathcal{M}\|_{\max}$ respectively. The transpose of matrix $\mathcal{M}$ is
denoted by $\mathcal{M}^\intercal$.  The complex conjugate of $c$ is denoted by
$c^*$.  The vectorization operation on a matrix corresponds to concatenating its
columns in order. The vectorized form of matrix $\mathcal{M} \in
\mathbb{C}^{N\times N}$ is denoted by the vector $\mathcal{M}^v \in
\mathbb{C}^{N^2}$.  The 2D-DFT operator is denoted \revv{by
$\mathfrak{F}(\cdot)$} such that $\g[k_1,k_2] = \mathfrak{F}(g[n_1,n_2])$ and
the inverse 2D-DFT (2D-IDFT) operator is denoted \revv{by
$\mathfrak{F}^{-1}(\cdot)$} such that $g[n_1,n_2] =
\mathfrak{F}^{-1}(\g[k_1,k_2])$. 
\section{Signal and sampling model}
\label{sec:model}

We consider grayscale images. \revv{As is natural for images, we consider that
the signal value at any pixel is in the interval $[0, 255]$}. The sampling is
done in the frequency domain. Let the image $g[n_1,n_2]$ be of $M \times M$
dimension. The 2D-DFT is computed such that there are $N\times N$ samples in the
2D-DFT of $g[n_1,n_2]$, i.e., in $\g[k_1,k_2]$. We ensure that $N>2M$ such that
the real part of 2D-DFT is sufficient to reconstruct the
image~\cite{so2018reconstruction}. 
\begin{figure*}[h]
\begin{center}
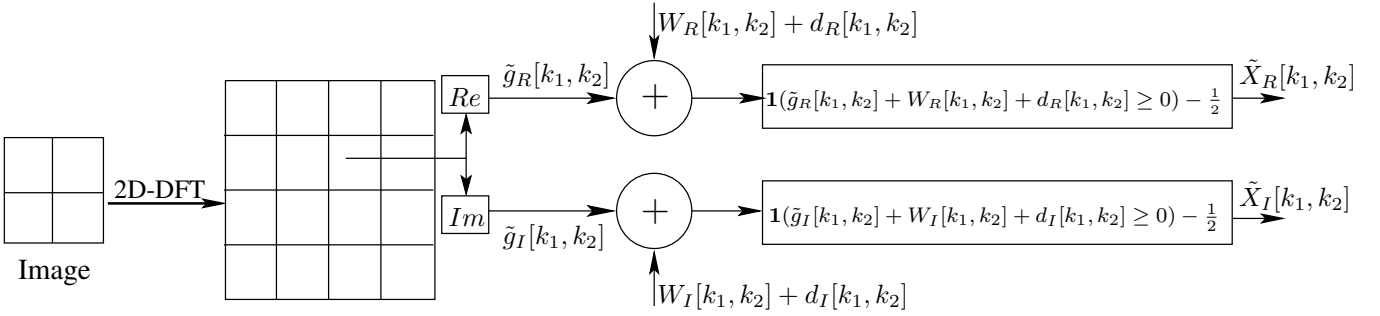
\caption{The sampling scheme for our reconstruction problem. The 2D-DFT of the
image has additive zero-mean noise, $W_R[k_1,k_2]$ and $W_I[k_1,k_2]$, in its
real and imaginary parts respectively. AWGN dither $d_R[k_1,k_2]$ and
$d_I[k_1,k_2]$ are also added. One bit each of the real and imaginary parts of
the noisy 2D-DFT are recorded via a comparator with $0$.}
\label{fig:sampling}
\end{center}
\end{figure*}
Denote the noise affecting the real and imaginary parts by $W_R[k_1,k_2]$ and
$W_I[k_1,k_2]$ respectively. For the noise model, we \revv{have the following
assumption:
\begin{assumption} \label{assumption:noise}
The real and imaginary parts of the 2D-DFT are affected by independent additive
zero-mean noise of a known distribution, which is symmetric w.r.t. $0$. 
%
%
\end{assumption}}

Additional additive white Gaussian noise (AWGN) of variance $\sigma_d^2$ is
added to dither the signal \revv{if the noise distribution doesn't have mass
over the entire region withing the bounds of the signal. This is required since
the proposed algorithm depends on the presence of noise of known distribution to
be able to recover the signal. This requirement is satisfied by adding the
dither noise. The MLE based algorithm of Bahmani, Boufounos, and Raj would also
require noise to be present, they handle it by heuristically modifying the cost
function when there is no noise~\cite{bahmani2013robust}.}

 Dither $d_R[k_1,k_2]$ and $d_I[k_1,k_2]$ are added to the real and imaginary
parts of the 2D-DFT respectively. The use of dithering with quantized signals is
well known \cite{wagdy1989effect}. In our model, dither is added to ensure that
there is sufficient variance in the samples, and the magnitude information is
captured in it. The use of dither is further discussed in
Section~\ref{subsec:dither}. In our sampling model, as shown in
Fig.~\ref{fig:sampling}, one-bit samples $\X_R[k_1,k_2]$ and $\X_I[k_1,k_2]$, of
the real part, $\tilde{g_R}[k_1,k_2]$, and imaginary part,
$\tilde{g_I}[k_1,k_2]$, respectively of the 2D-DFT, $\tilde{g}[k_1,k_2]$, are
recorded as follows: 
\begin{align}
 &\X_R[k_1,k_2] = \nonumber \\
 &\mathlarger{\mathbbm{1}}\left( \tilde{g_R}[k_1,k_2] + W_R[k_1,k_2] +
d_R[k_1,k_2] >0  \right ) - \frac{1}{2}  \label{eq:samp1}, \\
&\X_I[k_1,k_2] = \nonumber \\
&\mathlarger{\mathbbm{1}}\left(  \tilde{g_I}[k_1,k_2] + W_I[k_1,k_2] +
d_I[k_1,k_2] >0  \right ) - \frac{1}{2}.  \label{eq:samp2}
\end{align}
Due to conjugate symmetry of the 2D-DFT of real signals, we need to sample only
half of the entries of $\X_R[k_1,k_2]$ and $\X_I[k_1,k_2]$. It must noted that
Algorithm~\ref{algo:2} uses both $\X_R[k_1,k_2]$ and $\X_I[k_1,k_2]$, i.e., two
bits of the noisy \ul{phase of the}  2D-DFT. Whereas, Algorithm~\ref{algo:1}
uses only $\X_R[k_1,k_2]$ i.e., only one bit of the noisy \ul{phase of the}
2D-DFT of the image.

Recall that the 2D-DFT of $g \in \mathbb{C}^{N \times N}$ is given by:
\begin{align}
\tilde{g}[k,l] = \sum_{n_1=0}^{N-1} \sum_{n_2=0}^{N-1} g[n_1,n_2]
\exp\Big(-j2\pi\Big( \frac{kn_1}{N}+\frac{ln_2}{N}  \Big)\Big). \label{eq:2d-dft}
\end{align}
Whereas, the 2D-IDFT of $\tilde{g} \in \mathbb{C}^{N \times N}$ is given by:
\begin{align}
g[k,l] = \frac{1}{N^2} \sum_{n_1=0}^{N-1} \sum_{n_2=0}^{N-1} \tilde{g}[n_1,n_2] \exp\Big(j2\pi\Big( \frac{kn_1}{N}+\frac{ln_2}{N}  \Big)\Big). \label{eq:2d-idft}
\end{align}
In the following section we give the two variations of the proposed algorithm,
i.e., Algorithms~\ref{algo:1} and~\ref{algo:2}.

\section{Proposed Signal Reconstruction Algorithm}
\label{sec:algo}

In Subsection~\ref{subsec:algo1} we discuss Algorithm~\ref{algo:1}, which uses
one-bit recordings of the 2D-DFT \ul{phase}. 
\subsection{Using one-bit recordings of the noisy 2D-DFT phase}
\label{subsec:algo1}
Recall that $\X_R[k_1,k_2]$ is the real part of $\X[k_1,k_2]$ and is recorded
according to sampling model~\eqref{eq:samp1}. \edit{Define the following subset of the set of real numbers, $\S:= [-255N^2, 255N^2]$.} \edit{Let operator $\P:\mathbb{S}^{N\times N} \to  \mathbb{S}^{N\times N}$ be defined as:
\begin{align}
\P(\g) & :=  \dft{\textsc{Clip}(\textsc{Proj}(\idft{\g}))}. \label{eq:P}
\end{align}
See that operator $\P$ performs a series of four operations on the frequency domain input. First it computes a 2D-DFT to convert the argument to the spatial domain. Then it sets the pixels corresponding to the zero-padding to zero. This operation is denoted by \textsc{Proj}. Then it It clips the pixel values to the range $[0,255]$, i.e., it changes the entries that are greater than $255$ to $255$, and that are less than $0$ to $0$.  This operation is denoted by \textsc{Clip}. Finally it takes the 2D-DFT of the spatial domain
image to obtain the frequency domain image. The \textsc{Clip} operation ensures that $\P(\g)$ is in $\S^{N \times N}$ for all $\g$ in $\S^{N \times N}$.  Notice that every operation within $\P$ is a projection operation and therefore, $\P$ performs a projection operation. By the property of projection operations, $\P$ is non-expansive with respect to the Frobenious norm. We call this the \textsc{Non-Expansive} property of $\P$.}

 Let the map $\T:
\edit{\mathbb{S}^{N\times N} \to  \mathbb{S}^{N\times N}} $ be defined as:
\begin{align}
\!\!\T(\g) & := \P \left(\gamma \X_R +  \g - \gamma \left(
\F(\g) - \frac{1}{2} \right) \right),~~ \g \in \S^{N\times
N} \label{eq:T} \\
\gamma & \in \left(0,\frac{2}{f_{\mbox{\footnotesize max}}} \right).
\label{eq:gammaT}
\end{align}

The function $\F: \mathbb{R}^{N\times N} \to  \mathbb{R}^{N\times N}$ is the
cumulative distribution function (CDF) of the noise (including dither), applied
element-wise on the argument. Similarly, function $f: \mathbb{R}^{N\times N} \to
\mathbb{R}^{N\times N}$ is the probability distribution function (PDF) of the
noise, applied element-wise on the argument. The constant
$f_{\mbox{\footnotesize max}}$ is the maximum value of \ul{ $f(x)$ in $x \in
(-\infty, \infty)$.} \revv{the noise PDF for $x$ in the region within the bounds
of the signal.} Condition~\eqref{eq:gammaT} on parameter $\gamma$ is required to
ensure that $\T(\g)$ is a contraction mapping, w.r.t. to the Frobenius norm
i.e., it reduces the Frobenius norm of the difference with the fixed point of
the mapping. The one-bit of the 2D-DFT \ul{phase} i.e., the sign of the real
part of the 2D-DFT is recorded in $\X_R[k_1,k_2]$ according to the sampling
model~\eqref{eq:samp1}. 

\st{The operator $\P$ projects the frequency domain argument to the support region in the spatial domain. It first takes the 2D-IDFT of the argument to obtain an image in the spatial domain. Then it sets the pixels corresponding to the zero-padding to zero. It clips the pixel values to the range $[0,255]$, i.e., it changes the entries that are greater than $255$ to $255$, and that are less than $0$ to $0$. Then it takes the 2D-DFT of the spatial domain image to obtain the frequency domain image. The clipping operation, performed within $\P$ in the spatial domain, ensures that $\T(\g)$ is in $\S^{N \times N}$ for all $\g$ in $\S^{N \times N}$.}

It must be noted that the 2D-IDFT of the real part of the 2D-DFT gives two
copies of the image if appropriate zero-padding is done, i.e., $N>2M$. Recall
that the original spatial domain image is of dimension $M \times M$ pixels. It
is zero-padded such that the 2D-DFT gives a frequency domain image of dimension
$N \times N$ pixels. The detailed algorithm is given in Algorithm~\ref{algo:1}. 

\st{After convergence to the fixed point of the iteration in Algorithm~
we may also filter the image using a guided filtering of appropriate smoothing
parameter. In the numerical simulations, we found that the
guided filter can improve the reconstruction quality in most images, however,
the results we present in Section~
are without the use of the
guided filter with our algorithm.}

\begin{algorithm}
\SetAlgoLined

\KwIn{$(\X_R, ~\F, ~\varepsilon)$}
\KwOut{$G$}
 \DontPrintSemicolon $\G_0 = 0$ \Comment*[r]{initialization}

 \DontPrintSemicolon\Repeat {$k^* \text{such that~} \|\G_{k^*} - \G_{k^*-1}\|_F \leq \varepsilon$}
 {    $\G_{k+1} = \T(\G_k)$  \Comment*[r]{Contraction mapping} } 
 \DontPrintSemicolon $G = \mathfrak{F}^{-1}(\G_{k^*})$	\Comment*[r]{Converting to spatial domain}
%
\caption{Algorithm using one-bit noisy recordings of the 2D-DFT for signal reconstruction}
 \label{algo:1}
\end{algorithm}
The main idea behind Algorithm~\ref{algo:1} is that if the map $\T$ is a
contraction mapping, by Banach's contraction mapping principle, $\T$ has a
unique fixed point. Here the complete metric space on which $\T$ is defined is
$\mathbb{R}^{N\times N}$. We show in Section~\ref{sec:proofs} that for  large
enough $N/M$, the fixed point of $\T$ is a good estimate of the real part of the
2D-DFT of the image to be reconstructed. More specifically, we show that the
expected MSE between the fixed point of $\T$ and the real part of the 2D-DFT of
the original image is $O(M^2/N^2)$. We use Picard's iteration to reach the fixed
point of $\T$, starting from any finite point. We show in Lemma~\ref{lem:T} in
Appendix~\ref{appendix:A} that the convergence of the recursion to the fixed
point is guaranteed by choosing the parameter $\gamma$ as in
eq.~\eqref{eq:gammaT}.  To obtain a fast convergence to the fixed point,
$\gamma$ is set very close to, but less than $\frac{2}{f_{\mbox{\footnotesize
max}}}$.  A denoised version of the required image in the spatial domain can be
obtained from an estimate of the real part of its 2D-DFT by computing its
2D-IDFT if there is sufficient zero padding, i.e., $(N>2M)$.

In Subsection~\ref{subsec:algo2} we discuss the second variation of the proposed
algorithm, i.e., Algorithm~\ref{algo:2}, which is used for signal reconstruction
from two-bit recordings of its 2D-DFT.

%
%
%

\subsection{Using two-bit recordings of the noisy 2D-DFT phase}
\label{subsec:algo2}

Let $\g[k_1,k_2] := \tilde{g}_R[k_1,k_2] + j\tilde{g}_I[k_1,k_2] \in
\mathbb{C}^{N\times N}$ be the image in the frequency domain and $\X[k_1,k_2] :=
\X_R[k_1,k_2] + j\X_I[k_1,k_2] \in \mathbb{C}^{N\times N}$ be recorded according
to the sampling models~\eqref{eq:samp1} and~\eqref{eq:samp2}. \edit{Define the following subset of the set of complex numbers, $ \S_c:= \{(a+jb)|a, b \in [-255N^2, 255N^2]\}$.} 

\edit{Let operator $\Q:\S_c^{N\times N} \to  \S_c^{N\times N}$ be defined as:
\begin{align}
\Q(\g) & :=  \dft{\textsc{Clip}(\textsc{Proj}(\idft{\g}))}. \label{eq:Q}
\end{align}
Operator  $\Q$ is similar to operator $\P$ used in
Algorithm~\ref{algo:1}, but applies to complex matrices. It computes the 2D-IDFT of the frequency domain argument, projects it onto the support region in the spatial domain (\textsc{Proj}), clips the pixel values to lie in the range $[0,255]$ (\textsc{Clip}), and then computes the 2D-DFT to get the output in the frequency domain. Recall that the support region in the spatial domain is given by the pixels other than the zero-padding pixels of the image. The clipping operation, performed within $\Q$ in the spatial domain, ensures that $\Z(\g)$ is in $\S_c^{N \times N}$ for all $\g$ in $\S_c^{N \times N}$. Similar to $\P,$ every operation within $\Q$ is a projection operation and therefore, $\Q$ is a projection operator. By the property of projection operators, $\Q$ is non-expansive with respect to the Frobenious norm. We call this the \textsc{Non-Expansive} property of $\Q$.}

Define the map $\Z
: \edit{ \S_c^{N\times N} \to \S_c^{N\times N}}$ as:
\begin{align}
&\Z(\tilde{g})  =  \label{eq:Z}\\
&Q \left(\gamma (\X_R+j\X_I) +   \tilde{g} - \gamma \left(
\F(\tilde{g}_R)+j\F(\tilde{g}_I) - \frac{1+j}{2} \right) \right), \nonumber  \\
&\gamma \in \left(0,\frac{2}{f_{\mbox{\footnotesize max}}} \right).
\label{eq:gammaZ}
\end{align}
Recall that the function $\F: \mathbb{R}^{N\times N} \to  \mathbb{R}^{N\times
N}$ is the cumulative distribution function of the noise and
$f_{\mbox{\footnotesize max}}$ is the maximum value of \ul{ $f(x)$ in $x \in
(-\infty, \infty)$.} \revv{the noise PDF for $x$ in the region within the bounds
of the signal.} \st{The operator $\Q$ is similar to operator $\P$ used in
Algorithm~
 but  applies to complex matrices and projects to the
support region and dynamic range in the spatial domain. The operator $\Q$ computes the 2D-IDFT of
the frequency domain argument, projects it onto the support region in the
spatial domain, clips the pixel values to lie in the range $[0,255]$, and then computes the 2D-DFT to get the output in the frequency
domain. Recall that the support region in the spatial domain is given by the
pixels other than the zero-padding pixels of the image. The clipping operation, performed within $\Q$ in the spatial domain, ensures that $\Z(\g)$ is in $\S_c^{N \times N}$ for all $\g$ in $\S_c^{N \times N}$.}
\begin{algorithm}
\SetAlgoLined
\KwIn{$(\X_R,\X_I ~\F, ~\varepsilon)$}
\KwOut{$G$}
 \DontPrintSemicolon $\G_0 = 0$ \Comment*[r]{initialization}

 \DontPrintSemicolon\Repeat {$k^* \text{such that~} \|\G_{k^*} - \G_{k^*-1}\|_F \leq \varepsilon$}
 {    $\G_{k+1} = \Z(\G_k)$  \Comment*[r]{Contraction mapping} } 
 \DontPrintSemicolon $G = \mathfrak{F}^{-1}(\G_{k^*})$	\Comment*[r]{Converting to spatial domain}
 
%
\caption{Algorithm using two-bit noisy recordings of the 2D-DFT for signal
reconstruction}
\label{algo:2}
\end{algorithm}

If $\Z$ is a contraction mapping, then by the Banach's contraction mapping
theorem, it has a unique fixed point. $\Z$ is defined over the complete metric
space $\mathbb{C}^{N \times N}$. In Lemma~\ref{lem:Z} in
Appendix~\ref{appendix:B}, we show that $\Z$ is indeed a contraction mapping
with the Frobenius norm as the distance metric. We also show in
Section~\ref{sec:proofs} that this fixed point of $\Z$ is a good estimate of the
2D-DFT of the original image, with the expected MSE of $O(M^2/N^2)$. As in
Algorithm~\ref{algo:1}, we use Picard's iteration to reach the fixed point of
$\Z$. As shown in Lemma~\ref{lem:Z}, the convergence of the recursion in
Algorithm~\ref{algo:2} is guaranteed by choosing the parameter $\gamma$ in $\Z$
as in eq.~\eqref{eq:gammaZ}.

\section{Theoretical Result on the Error in Reconstruction}
\label{sec:proofs}
In this section we provide proofs of the bounds on the expected MSE for the two
variations of the algorithm given in the previous section. In
Subsection~\ref{subsec:proof1}, we give a proof of the error bound for
Algorithm~\ref{algo:1}.
\subsection{Error bound for Algorithm~\ref{algo:1}}
\label{subsec:proof1}
In this sub-section, we give the error convergence result for
Algorithm~\ref{algo:1} as a function of $M/N$. Let the original spatial domain
image be $g[n_1,n_2]$ and its 2D-DFT be given by $\g[k_1,k_2] =
\g_R[k_1,k_2]+j\g_I[k_1,k_2]$. Define $l_R := \F(\g_R)- \frac{1}{2}$. Recall
that the \ul{one-bit recordings of the noisy phase of} \revv{sign recordings of
noisy} $\g[k_1,k_2]$ are $\X_R[k_1,k_2]$. \st{Notice that $\mathbb{E}[\X_R] = l_R$.} Define $S := \P(\X_R)$. \st{Recall that $\gamma S$ denotes the first term in the map $\T$.}
\edit{
\begin{lemma} \label{lem:exp-S}
$\mathbb{E}[\X_R]=l_R$ and $\mathbb{E}[S]=\P(l_R).$
\end{lemma}
\begin{proof}
$\mathbb{E}[\X_R]=l_R$ follows from the definition of $\X_R$ and the fact that the noise is additive zero-mean and its distribution is symmetric w.r.t. $0$ (Assumption~\ref{assumption:noise}). Recall that $\P$ performs a series of four operations: $\idft{}, \textsc{Proj}, \textsc{Clip},$ and $\dft{}$. When applying \textsc{Clip} to $\textsc{Proj}(\idft{\X_R})$, the argument remains unchanged. This is because $\X_R[k_1,k_2] \in \{\frac{-1}{2}, \frac{1}{2}\}$ and by the definition of $\idft{}$ in eq.~\eqref{eq:2d-idft}, each element of the $\textsc{Proj}(\idft{\X_R})$ is in $[0,255].$\footnote{$\X_R$ is pre-processed to preserve symmetry due to which $\idft{\X_R}\geq 0.$} Therefore, $\P(\X_R) = \dft{\textsc{Proj}(\idft{\X_R})}$ and $\P(l_R) = \dft{\textsc{Proj}(\idft{l_R})}.$ See that $\dft{\textsc{Proj}(\idft{\cdot})}$ is a linear operation, and by the linearity of expectation, we have $\mathbb{E}[S] = \P(\mathbb{E}[\X_R]) = \P(l_R)$.
\end{proof}}
 Consider the following recursion using the map $\T$:
%
\begin{align}
{\G}_{0} = 0,~~ {\G}_{k+1} = \T({\G}_{k}).\nonumber
\end{align}
Let the fixed point of this recursive mapping be ${\G}_{\mbox{\footnotesize
one-bit}}$. We now derive a bound on $\frac{1}{N^2} \mathbbm{E}[ \|
{\G}_{\mbox{\footnotesize one-bit}} - \g_R \|_F^2]$, i.e., the mean-squared
error in ${\G}_{\mbox{\footnotesize one-bit}}$ as an estimate of $\g_R$. The
following theorem is the main result with respect to Algorithm~\ref{algo:1}. The
required lemmas are in Appendix~\ref{appendix:A}.
\begin{theorem} \label{thm:algo1}
The expected MSE for Algorithm~\ref{algo:1} is $O(M^2/N^2)$.
\end{theorem}
\begin{proof}
Recall the definition $S := \P(\X_R)$ and the map $\T$ in eq.~\eqref{eq:T}.
Consider two recursions, one using $\gamma S$ in $\T$ and having
$\G_{\mbox{\footnotesize one-bit}}$ as fixed point and the other using $\gamma
\P(l_R)$ in $\T$ and having $\g_R$ as fixed point. Note that the first
recursion, using $\gamma S$, corresponds to using one-bit noisy samples of
$\g_R[k_1,k_2]$, whereas the second recursion, using $\gamma \P(l_R)$,
corresponds to using the perfect information of $\g_R[k_1,k_2]$. Since the
recursion is insensitive to initialization, we start it with $0$. Let,

\begin{align}
\G_{0} &= \g_{0} = 0, \nonumber\\
\G_{k+1} &= \T(\G_{k}) =  \gamma S + \P\left(  \G_{k} - \gamma  \left(
\F({\G}_{k})-\frac{1}{2}  \right)  \right),  \label{eq:recursion1}\\
\g_{k+1} &= \gamma \P(l_R) + \P  \left(   \g_{k} - \gamma  \left(
\F(\g_{k})-\frac{1}{2} \right)   \right).  \label{eq:recursion2}
\end{align}
The distortion in the reconstructed image in the frequency domain is captured in
the difference between these two recursions. Consider the following difference,
\begin{align}
%
%
& \G_{k+1} - \g_{k+1} = \nonumber \\
& \gamma(S-\P(l_R)) + \P \left( {\G}_{k}-\g_{k} -\gamma \left(
\F({\G}_{k})-\F(\g_{k})\right) \right).\nonumber 
\end{align}
Using the triangle inequality for the Frobenius norm
\cite{kreyszig1989introductory},
\begin{align}
&\left\| \G_{k+1} - \g_{k+1} \right\|_F \leq \gamma \left\|S-\P(l_R)  \right\|_F
\nonumber\\
& + \left \|\P \left( {\G}_{k}-\g_{k} -\gamma \left(
\F({\G}_{k})-\F(\g_{k})\right) \right)  \right\|_F.\nonumber
\end{align}
By the \textsc{Non-Expansive} property of $\P$,
\begin{align}
&\left\| \G_{k+1} - \g_{k+1} \right\|_F\nonumber\\
&  \leq  \gamma \left\|S-\P(l_R)  \right\|_F  + \left \|{\G}_{k}-\g_{k} -\gamma
\left( \F({\G}_{k})-\F(\g_{k})\right)  \right\|_F. \nonumber
\end{align}
By the Lagrange mean value theorem, $\F({\G}_{k})-\F(\g_{k}) = f(c_k)({\G}_{k} -
\g_{k})$ for some matrix $c_k \in \mathbb{R}^{N \times N}$ between ${\G}_{k}$
and $\g_{k}$ such that each entry of $c_k$ is between the corresponding entries
of ${\G}_{k}$ and $\g_{k}$. \edit{Since ${\G}_{k}$ and $\g_{k}$ are outputs of $T,$ they are both in the bounded set $\mathbb{S}^{N\times N}$ and therefore $c_k$ is also in  $\mathbb{S}^{N\times N}$.} Using this, we get
\begin{align}
&\left\| \G_{k+1} - \g_{k+1} \right\|_F \leq \nonumber\\
& \gamma \left\|S-\P(l_R)  \right\|_F  + \left \| \left(1-\gamma f(c_k)\right)
\left({\G}_{k}-\g_{k}\right) \right\|_F. \nonumber
\end{align}
Define $\alpha$ as: $\alpha = \|1-\gamma f \|_{\max}$. Using $\alpha$, we get
\begin{align}
& \left\| \G_{k+1} - \g_{k+1} \right\|_F \leq \gamma \|S-\P(l_R)\|_F  + \alpha
\left\|\left( \G_{k} - \g_{k} \right) \right\|_F, \nonumber\\
& \text{or,~}  \left\| \G_{k+1} - \g_{k+1} \right\|_F - \alpha \left\|\G_{k} -
\g_{k} \right\|_F \leq \gamma \|S-\P(l_R)\|_F.\nonumber
\end{align}
For parameter $\gamma$ chosen according to eq.~\eqref{eq:gammaT}, we know from
Lemma~\ref{lem:T} that both recursions in eqs.~\eqref{eq:recursion1}
and~\eqref{eq:recursion2} converge to their respective fixed points. Thus the
following holds,
\begin{align}
 \lim_{k\to\infty} \left\| \G_k - \g_k \right\|_F - \alpha \left\|\G_{k-1} -
\g_{k-1} \right\|_F &\leq \gamma \|S-\P(l_R)\|_F, \nonumber\\
 \text{or,~} (1-\alpha) \left\| \G_{\mbox{\footnotesize one-bit}} - \g_R
\right\|_F &\leq \gamma \left\|S-\P(l_R) \right \|_F.\nonumber
\end{align}
Squaring  and taking expectation on the previous inequality,
\begin{align}
\mathbbm{E}[ \| {\G}_{\mbox{\footnotesize one-bit}} - \g_R \|_F^2] &\leq
\frac{\gamma^2 \mathbbm{E}[\|S-\P(l_R)\|_F^2]}{(1-\alpha)^2}.\nonumber
\end{align}
Since we want to calculate the expected mean squared error, we average over all
entries of the estimated image in the frequency domain. \edit{By Lemma~\ref{lem:exp-S} and} the bound on the
variance of $S$ from Lemma~\ref{lem:var1} in Appendix~\ref{appendix:A}, we get
\begin{align}
\frac{1}{N^2} \mathbbm{E}[ \| {\G}_{\mbox{\footnotesize one-bit}} - \g_R \|_F^2]
\leq \frac{\gamma^2 (M^2/2N^2)}{(1-\alpha)^2}.\label{eq:result:thm1}
\end{align}
\revv{See that the addition of AWGN dither noise ensures that $f(x) > 0$ for all
$x$ in the region within the bound of the signal. This ensures that for $\gamma$
chosen via eq.~\eqref{eq:gammaT}, $\alpha <1.$ For example with $\gamma =
1/f_{\max}$ we get $\alpha = 1 -f_{\min}/f_{\max} .$} Since the parameters
$\gamma$ and $\alpha$ are independent of the image dimension, \ul{we obtain
that} the expected MSE in the estimate of $\g_R[k_1,k_2]$ is $O(M^2/N^2)$. To
obtain an estimate of spatial domain image $g[n_1,n_2]$, we compute the 2D-IDFT
of $\G_{\mbox{\footnotesize one-bit}}$. Now we get a $N \times N$ pixels image
in the spatial domain which consists of two copies of the required estimate of
$g[n_1,n_2]$ of $M \times M$ pixels each, surrounded by zero-padding. We take an
average over these two copies to get the final estimate. \ul{Since the 2D-IDFT
is a linear transform, the expected MSE in the spatial domain is also
$O(M^2/N^2)$.}
\end{proof}

\subsection{Error bound for Algorithm~\ref{algo:2}}
In this subsection we give the proof for the variation of the algorithm using
two-bit recordings of the noisy \ul{phase of the frequency domain image}
\revv{2D-DFT. The proof is very similar to that for Algorithm~\ref{algo:1}. The
only major difference is that here we deal with complex numbers instead of real
numbers.} We need the following definitions:
\begin{align}
Y &:= \Q(\X_R+j\X_I),\\
l &:= \F(\g_R) + j\F(\g_I) -\frac{1+j}{2}.
\end{align}
See that $Y$ is analogous to $S$ in the proof of Theorem~\eqref{thm:algo1}.
$\gamma Y$ corresponds to the first term in the map $\Z$ in eq.~\eqref{eq:Z}.

\edit{
\begin{lemma} \label{lem:exp-Y}
$\mathbb{E}[\X]=l$ and $\mathbb{E}[Y]=\Q(l).$
\end{lemma}
\begin{proof}
The proof is similar to that of Lemma~\ref{lem:exp-S}. $\mathbb{E}[\X]=l$ follows from Assumption~\ref{assumption:noise} on the noise distribution. 
\textsc{Clip} does not change $\textsc{Proj}(\idft{\X})$. This is because each element of $\textsc{Proj}(\idft{\X})$ is in $[0,255].$ Therefore, $\Q(\X) = \dft{\textsc{Proj}(\idft{\X})}$ and $\Q(l) = \dft{\textsc{Proj}(\idft{l})}.$ Since $\dft{\textsc{Proj}(\idft{\cdot})}$ is a linear operation, and by the linearity of expectation, we have $\mathbb{E}[Y] = \P(\mathbb{E}[\X]) = \P(l_R)$.
\end{proof}
}

Consider the following recursion using $\Z$\st{ as in Algorithm 2} 
\begin{align}
 \G^{(0)} = 0, ~~ \G^{(k+1)} = \Z({\G}^{(k)}).\nonumber
\end{align}
Let the fixed point of this recursive mapping be ${\G}_{\mbox{\footnotesize
two-bit}}$. We now derive a bound on $\frac{1}{N^2} \mathbbm{E}[ \|
{\G}_{\mbox{\footnotesize two-bit}} - \g\|_F^2]  $.
The following theorem is the main result with respect to Algorithm~\ref{algo:2}.
The required lemmas are in the Appendix~\ref{appendix:B}.
\begin{theorem} \label{thm:algo2}
The expected MSE for Algorithm~\ref{algo:2} is $O(M^2/N^2)$.
\end{theorem}
\begin{proof}
\ul{The proof is on similar lines as that for Theorem}~\eqref{thm:algo1}.
Consider two recursions, one using $\gamma Y$ and having
$\G_{\mbox{\footnotesize two-bit}}$ as its fixed point and the other using
$\Q(\l)$, having $\g$ as its fixed point. Note that the first recursion uses the
two-bit noisy recordings of the \ul{phase of the frequency domain image} 2D-DFT,
whereas the second recursion uses the perfect information of the original
frequency domain image $\g[k_1,k_2]$. Let,
\begin{align}
&\G^{(0)} = \g^{(0)} = 0,  \nonumber\\
&\G^{(k+1)} = \Z(\G^{(k)}) = \label{eq:algo2rec1}\\
& \gamma Y + \Q \left(\G^{(k)} - \gamma \left( \F(\G_R^{(k)}) + j\F(\G_I^{(k)})
-\frac{1+j}{2} \right) \right),  \nonumber  
\end{align}
\begin{align}
& \g^{(k+1)}  = \label{eq:algo2rec2}\\
& \gamma \Q(l) +
\Q \left(\g^{(k)} - \gamma \left( \F(\g_R^{(k)}) + j\F(\g_I^{(k)})
-\frac{1+j}{2} \right) \right).  \nonumber
\end{align}
To calculate the distortion, consider the following difference,
\begin{align}
& \G^{(k+1)}  - \g^{(k+1)} = \gamma(\Y-\Q(\l))   + \Q \left(\G^{(k)}  -
\g^{(k)}\right) \nonumber \\ & - \gamma \Q \left( \F(\G_R^{(k)}) +
j\F(\G_I^{(k)}) - \F(\g_R^{(k)}) - j\F(\g_I^{(k)}) \right).\nonumber
 %
%
\end{align}
Using the triangular inequality for Frobenius norm
\cite{kreyszig1989introductory},

\begin{align}
& \left\| \G^{(k+1)}  - \g^{(k+1)} \right\|_F \leq  \gamma \left\|\Y - \Q(\l)
\right\|_F + \Big \| \Q \left(\G^{(k)}  - \g^{(k)}\right) \nonumber\\
& - \gamma \Q \left( \F(\G_R^{(k)}) + j\F(\G_I^{(k)}) - \F(\g_R^{(k)}) -
j\F(\g_I^{(k)}) \right) \Big\|_F.\nonumber
\end{align}

By the \textsc{Non-Expansive} property of $\Q$,
\begin{align}
\left \| \G^{(k+1)}  - \g^{(k+1)} \right\|_F  \leq  \gamma \left\|\Y - \Q(\l)
\right\|_F + \Big \| \left(\G^{(k)}  - \g^{(k)}\right)\nonumber \\
- \gamma \left( \F(\G_R^{(k)}) + j\F(\G_I^{(k)}) - \F(\g_R^{(k)}) -
  j\F(\g_I^{(k)}) \right) \Big\|_F. \nonumber
\end{align}
By  the Lagrange mean value theorem, $\F(\G_R^{(k)})-\F(\g_R^{(k)}) =
f(a_k)(\G_R^{(k)} - \g_R^{(k)})$ for some matrix $a_k \in \mathbb{R}^{N \times
N}$ between $\G_R^{(k)}$ and $\g_R^{(k)}$ such that each entry of $a_k$ is
between the corresponding entries of $\G_R^{(k)}$ and $\g_R^{(k)}$. Similarly,
$\F(\G_I^{(k)})-\F(\g_I^{(k)}) = f(b_k)(\G_I^{(k)} - \g_I^{(k)})$ for some
matrix $b_k \in \mathbb{R}^{N \times N}$ between $\G_I^{(k)}$ and $\g_I^{(k)}$. \edit{Since ${\G}_{k}$ and $\g_{k}$ are outputs of $Q,$ they are both in the bounded set $\S_c^{N\times N}$ and therefore $a_k + jb_k$ is also in $\S_c^{N\times N}$. Using this:}
\begin{align}
 & \left \| \G^{(k+1)}  - \g^{(k+1)} \right\|_F \leq  \gamma \left\|\Y - \Q(\l)
\right\|_F + \nonumber \\
& \big \| (1-\gamma f(a_k))\big(\G_R^{(k)}  - \g_R^{(k)}\big) +  j(1-\gamma
f(b_k))\big(\G_I^{(k)}  - \g_I^{(k)}\big)  \big\|_F. \nonumber
\end{align}
Using $\alpha = \|1-\gamma f \|_{\max}$, we get,
\begin{align}
&\left \| \G^{(k+1)}  - \g^{(k+1)} \right\|_F  \nonumber \\ 
&  \leq \gamma \|\Y-\Q(\l)\|_F + \alpha \left\|\left(\G_R^{(k)} -
\g_R^{(k)}\right) + j\left( \G_I^{(k)} - \g_I^{(k)} \right)\right \|_F,
\nonumber    \\
&\leq   \gamma \|\Y-\Q(\l)\|_F  + \alpha \left\|\G^{(k)} - \g^{(k)}
\right\|_F,\nonumber \\
& \text{or,~}\nonumber\\
& \left\| \G^{(k+1)}  - \g^{(k+1)} \right\|_F -\alpha \left\| \G^{(k)} -
\g^{(k)} \right\|_F \leq \gamma \|\Y-\Q(\l)\|_F.\nonumber
\end{align}
For parameter $\gamma$ chosen according to eq.~\eqref{eq:gammaZ}, from
Lemma~\ref{lem:Z} we know that both the recursions in eqs.~\eqref{eq:algo2rec1}
and \eqref{eq:algo2rec2} converge to their respective fixed points. Thus,
\begin{align}
&  \lim_{k\to\infty} \left\| \G^{(k+1)}  - \g^{(k+1)} \right\|_F -\alpha \left\|
\G^{(k)} - \g^{(k)} \right\|_F\nonumber\\
&   \leq \gamma \| \Y-\Q(\l) \|_F, \nonumber \\
& \text{or,~} (1-\alpha) \left\| \G_{\mbox{\footnotesize two-bit}} - \g
\right\|_F \leq  \gamma \left\| \Y-\Q(\l) \right \|_F. \nonumber 
\end{align}
Squaring  and taking expectation on the above inequality,
\begin{align}
\mathbbm{E}[ \| {\G}_{\mbox{\footnotesize two-bit}} - \g\|_F^2] &\leq
\frac{\gamma^2 \mathbbm{E}[\|\Y-\Q(\l)\|_F^2]}{(1-\alpha)^2}.\nonumber
\end{align}
Since we want to derive the expected mean squared error, we average over all
entries of the estimated image. Using \edit{Lemma~\ref{lem:exp-Y} and} the bound on the variance of $Y$ from Lemma~\ref{lem:var2} in Appendix~\ref{appendix:B},
\begin{align}
 \frac{1}{N^2} \mathbbm{E}[ \| {\G}_{\mbox{\footnotesize two-bit}} - \g\|_F^2]
\leq \frac{\gamma^2 (M^2/2N^2)}{(1-\alpha)^2}. \label{eq:result:thm2}
\end{align}
Since the parameters $\gamma$ and $\alpha$ are independent of the image
dimensions $M$ and $N$, we obtain that the expected MSE in the estimate of
$\g[k_1,k_2]$ is $O(M^2/N^2)$. To obtain an estimate of the image $g[n_1,n_2]$,
we compute the 2D-IDFT of ${\G}_{\mbox{\footnotesize two-bit}}$. \ul{Since the
2D-IDFT is a linear transform, the expected MSE in reconstruction in the spatial
domain is also $O(M^2/N^2)$.} This completes the proof.
%
%
\end{proof}

\subsection{The use of dither}
\label{subsec:dither}
In the sampling model~\eqref{eq:samp1} and~\eqref{eq:samp2}, we mentioned that
we add dither to the signal \ul{before recording the one-bit samples via the
comparator with $0$.} \revv{if there is insufficient noise.} From the results
for Algorithm~\ref{algo:1} and Algorithm~\ref{algo:2} in~\eqref{eq:result:thm1}
and~\eqref{eq:result:thm2} respectively, we see that the error bound depends on
$\gamma$ and $\alpha$, which depend on the noise distribution. Recall that
$\alpha = \|1-\gamma f \|_{\max}$. Notice that if the noise is not present, then
$\alpha = 1$, and the error bound tends to infinity. Further, for a very large
noise variance, $\alpha$ is close to $1$ and the error bound is large.
\revv{This justifies the use of dither when the noise is small, and also
explains why dither is not required if there is significant noise with the
signal.} \ul{The optimal value of the noise variance is between these two
extremes, and the dither is used to get close to this noise level.}

\section{Numerical Simulations}
\label{sec:sim}
In this section we provide the numerical validation of the results in
Theorems~\ref{thm:algo1} and~\ref{thm:algo2}. We compare Algorithm~\ref{algo:1}
and Algorithm~\ref{algo:2} with two state-of-the-art methods. The first
\ul{state-of-the-art} \revv{algorithm} we consider is the iterative algorithm
given in \cite{curtis1985signal}. We call it the `COL' algorithm after the
initials of the authors, Curtis, Oppenheim, and Lim. The algorithm COL requires
an initial estimate of the 2D-DFT magnitude. As in \cite{curtis1985signal}, we
use an average of 2D-DFT magnitudes of a large number of natural images to
provide this estimate. The other \ul{state-of-the-art} \revv{algorithm} is
derived from the compressed sensing algorithm: \emph{binary iterative
hard-thresholding with partial support estimate weighting (BIHT-PSW)}
\cite{north2015one}. It is shown in \cite{north2015one} that their algorithm is
robust against noise. The algorithm makes use of the knowledge of the support
region. In the scenario of this paper, we have complete knowledge of the support
in the spatial domain and thus we call it \emph{binary iterative
hard-thresholding with support information, or BIHT-SI}. \edit{We do not use convex relaxation based methods \cite{boufounos20081, zymnis2010compressed} for the comparison because they are not designed for the relatively high noise variance regime we consider in this paper.}
 
 We consider four benchmark images, viz, the cameraman, Lena, peppers, and IEEE
logo images of dimensions $128\times128$ pixels each. Thus, $M=128$ in this
experiment. These four original images are shown in Fig.~\ref{orig}. The noise
\ul{(including dither)} is considered to be uncorrelated AWGN of variance
$\sigma^2=100$ \ul{, added separately to the real part and imaginary part of the
2D-DFT.} \revv{We do not add dither in the experiments.}
 
\textbf{Experiment 1.} In this test, we verify \ul{the results of
Theorems}~\ref{thm:algo1} \ul{and }\ref{thm:algo2}\ul{ i.e.,} that the expected
MSE is asymptotically of order $O(M^2/N^2)$. Here the oversampling ratio is
$N^2/M^2$ for the image.  The pixel intensities in the spatial domain are in the
range of $[0,255]$. For natural images, on computing the 2D-DFT, the intensities
are much larger for the lower frequencies and smaller for the higher
frequencies. The value of the noise variance relative to the pixel intensities
is \ul{moderately} high, as evident from the weak reconstruction performance of
the \emph{BIHT-SI} and \emph{COL} despite the use of guided filtering too, as
shown in Figs.~\ref{oppen4} and~\ref{cs4}.

 The results of $\log_{10}(MSE)$ v/s $\log_{10}(N/M)$ for fixed $M$ are shown in
Figs.~\ref{fig:sub-first}, \ref{fig:sub-second}, \ref{fig:sub-third}, and
\ref{fig:sub-fourth}. The slopes of the curves in Figs.~\ref{fig:sub-first},
\ref{fig:sub-second}, \ref{fig:sub-third}, and \ref{fig:sub-fourth},
corresponding to Algorithm~\ref{algo:1}, and Algorithm~\ref{algo:2} are given in
Table~\ref{table:slopes}. This verifies Theorems~\ref{thm:algo1}
and~\ref{thm:algo2}. Whereas, the slopes corresponding to the other two
algorithms, \emph{BIHT-SI} and \emph{COL}, are much smaller and diminishing for
larger values of the sampling rate. This implies that their expected MSE doesn't
improve considerably with a higher sampling rate.

\textbf{Experiment 2.} In this test, we compare the quality of image
reconstruction by the four algorithms being considered. The PSNR, popular due to
its simplicity, is generally not a good metric to compare images
\cite{guan2006edge}. Unlike the human visual system (HVS), the PSNR doesn't
consider structural information such as the edges of the image. The structural
similarity index (SSIM)\cite{wang2004ssim} performs better than the PSNR in this
regard. However, it fails to measure the badly blurred images. More recent
techniques, the Edge Based Structural Similarity (ESSIM) and the Multi-Scale
Structural Similarity (MS-SSIM) \cite{wang2003multi} are designed to improve
upon the SSIM in this regard. The MS-SSIM uses dyadic wavelet transform instead
of Sobel filtering as in the SSIM. We employ these four metrics: PSNR, SSIM,
ESSIM, and MS-SSIM, for the reconstruction quality in this experiment.

\ul{For this experiment, The noise is considered to be uncorrelated AWGN of
variance $\sigma^2=100$, added separately to the real and imaginary parts of the
2D-DFT. The variance is assumed to be known.} For Algorithm~\ref{algo:1},
\emph{BIHT-SI}, and \emph{COL} the one-bit signal recording is of size
$2048\times 2048$ pixels of the noisy real part of the 2D-DFT of the image. For
Algorithm~\ref{algo:2}, the two-bit signal recording is of size $1448\times
1448$ pixels and has one bit each of the noisy real and imaginary parts of the
2D-DFT of the image. Note that $1448 \approx 2048/\sqrt{2}$ and thus, there are
equal number of bits of information and a fair comparison between the four
algorithms. 

The reconstructed images with \emph{COL},  \emph{BIHT-SI},
Algorithm~\ref{algo:1}, and Algorithm~\ref{algo:2} are shown in
Figs.~\ref{oppen4},~\ref{cs4},~\ref{b1}, and~\ref{b2} respectively. Tables
\ref{table:lena}, \ref{table:logo}, \ref{table:cm}, and \ref{table:pep} have the
reconstruction quality metrics for these four algorithms\footnote{Note that the
output of the ESSIM and MS-SSIM depends on calibration parameters. We have set
these parameters to have good contrast in the results for the four algorithms.}.
It can be observed that Algorithm~\ref{algo:1} and Algorithm~\ref{algo:2} have a
similar performance and they produce better results than \emph{BIHT-SI} and
\emph{COL} on all images by all four metrics of image reconstruction quality.
\begin{table}[hbt!]
%
%
\begin{tabular}{|P{0.07\textwidth}||P{0.05\textwidth}|P{0.09\textwidth}|P{0.09\textwidth}|P{0.06\textwidth}|}
%
%
%
\hline
Image & Lena & IEEE logo & Cameraman & Peppers\\
\hline
Algo.~\ref{algo:1}    &  2.03  &   2.01   &  2.00  &   2.04\\
Algo.~\ref{algo:2}    &  2.01  &   1.97   &   1.99  &   2.02\\ 
\hline
\end{tabular}
\caption{Asymptotic slopes of $\log_{10}(MSE)$ v/s $\log_{10}(N/M)$ plots. }
\label{table:slopes}
\end{table}
%
%
\begin{figure*}[p]
\begin{subfigure}[b]{\textwidth}
\centering
\includegraphics[width=.22\linewidth]{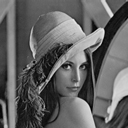}
\hfill
\includegraphics[width=.22\linewidth]{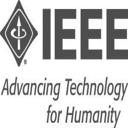}
\hfill 
\includegraphics[width=.22\linewidth]{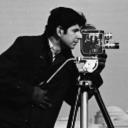}
\hfill 
\includegraphics[width=.22\linewidth]{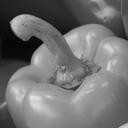}
\caption{Undistorted original images: Lena, IEEE logo, cameraman, and peppers.
All four images are of $128 \times 128$ pixels.}
\label{orig}
\end{subfigure}
%
%
%
\begin{subfigure}{\textwidth}
\centering
\includegraphics[width=.22\linewidth]{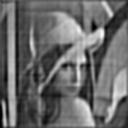}	
\hfill 
\includegraphics[width=.22\linewidth]{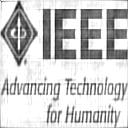}
\hfill 
\includegraphics[width=.22\linewidth]{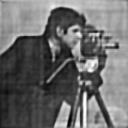}
\hfill 
\includegraphics[width=.22\linewidth]{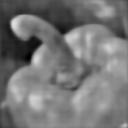}
\caption{Images reconstructed using Algorithm COL with a 2D-DFT of size $2048
\times 2048$ and noise variance equal to $100$. The images were sharpened and
passed thorough a guided filter after reconstruction using \emph{COL} to obtain
the best possible PSNR.} 
\label{oppen4}
\end{subfigure}
%
%
%
%
\begin{subfigure}{\textwidth}
\centering
\includegraphics[width=.22\linewidth]{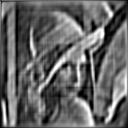}
\hfill 
\includegraphics[width=.22\linewidth]{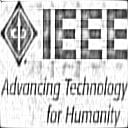}
\hfill 
\includegraphics[width=.22\linewidth]{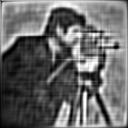}
\hfill 
\includegraphics[width=.22\linewidth]{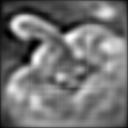} 
\caption{Images reconstructed using Algorithm BIHT-SI with a 2D-DFT of size
$2048 \times 2048$ and noise variance equal to $100$. The images were sharpened
and passed thorough a guided filter after reconstruction using \emph{BIHT-SI} to
obtain the best possible PSNR.}
\label{cs4}
\end{subfigure}
\begin{subfigure}{\textwidth}
\centering
\includegraphics[width=.22\linewidth]{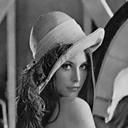}
\hfill 
\includegraphics[width=.22\linewidth]{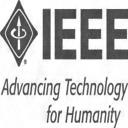}
\hfill 
\includegraphics[width=.22\linewidth]{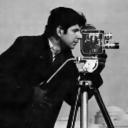}
\hfill 
\includegraphics[width=.22\linewidth]{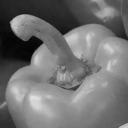}
\caption{Images reconstructed using Algorithm~\ref{algo:1} with a 2D-DFT of size
$2048 \times 2048$ and noise variance equal to $100$.}
\label{b1}
\end{subfigure}
%
%
%
%
%
\begin{subfigure}{\textwidth}
\centering
\includegraphics[width=.22\linewidth]{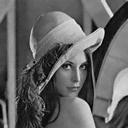}
\hfill 
\includegraphics[width=.22\linewidth]{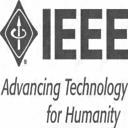}
\hfill 
\includegraphics[width=.22\linewidth]{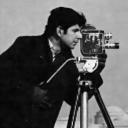}
\hfill 
\includegraphics[width=.22\linewidth]{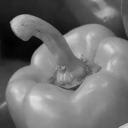}
\caption{Images reconstructed using Algorithm~\ref{algo:2} with a 2D-DFT of size
$1448 \times 1448$ and noise variance equal to $100$.}
\label{b2}
\end{subfigure}
%
%
%
\caption{Simulation results on four benchmark images. The The PSNR and other
reconstruction quality metrics are in
Tables~\ref{table:lena},~\ref{table:logo},~\ref{table:cm}, and~\ref{table:pep}
respectively in the left to right order of the images.}
\end{figure*}
%
%
\begin{figure*}[t]
\begin{subfigure}{.5\textwidth}
\centering
%
%
\includegraphics[height= 0.6\linewidth, width=.85\linewidth]{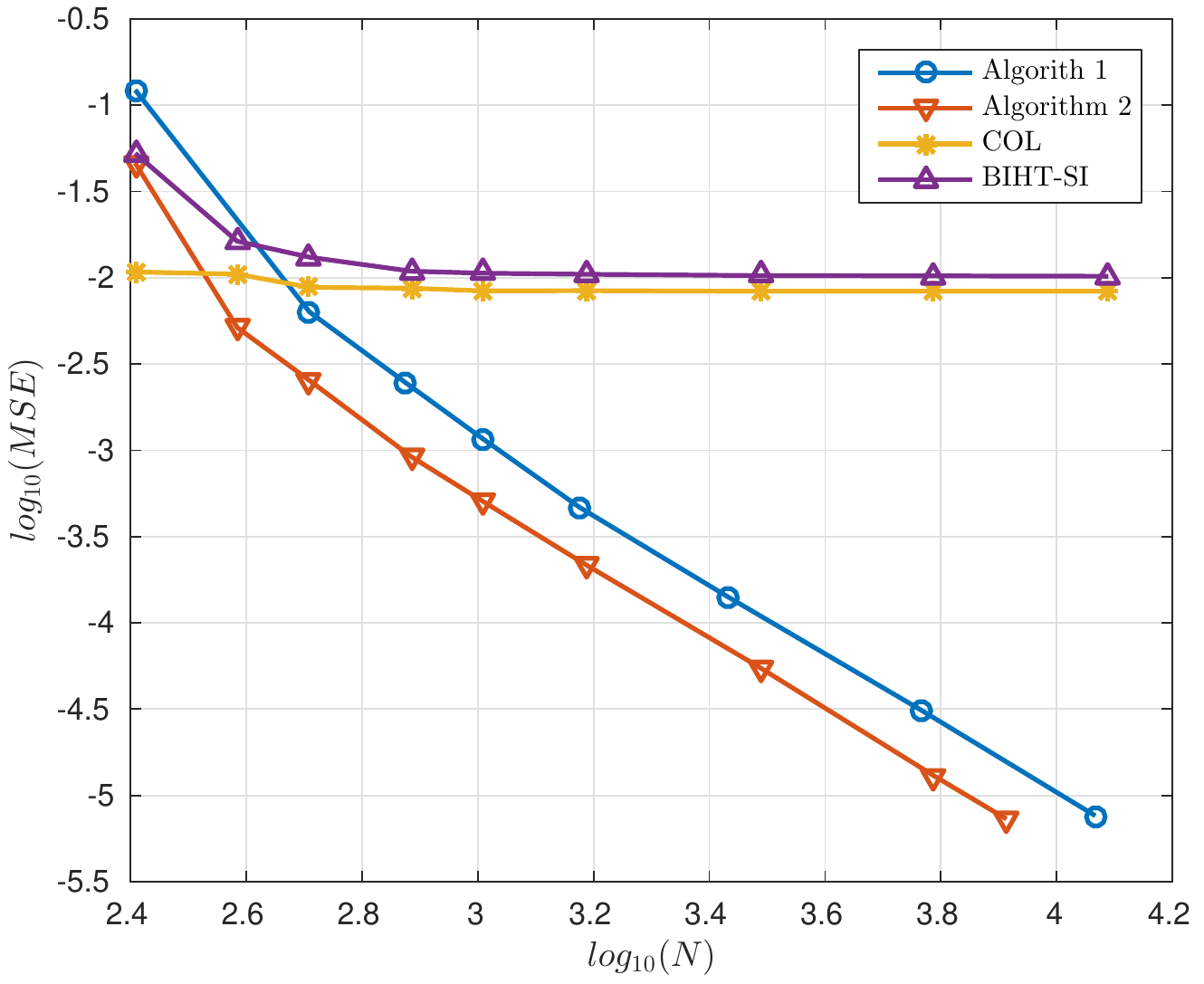}  
\caption{Lena}
\label{fig:sub-first}
\end{subfigure}
\begin{subfigure}{.5\textwidth}
\centering
%
%
\includegraphics[height= 0.6\linewidth, width=.85\linewidth]{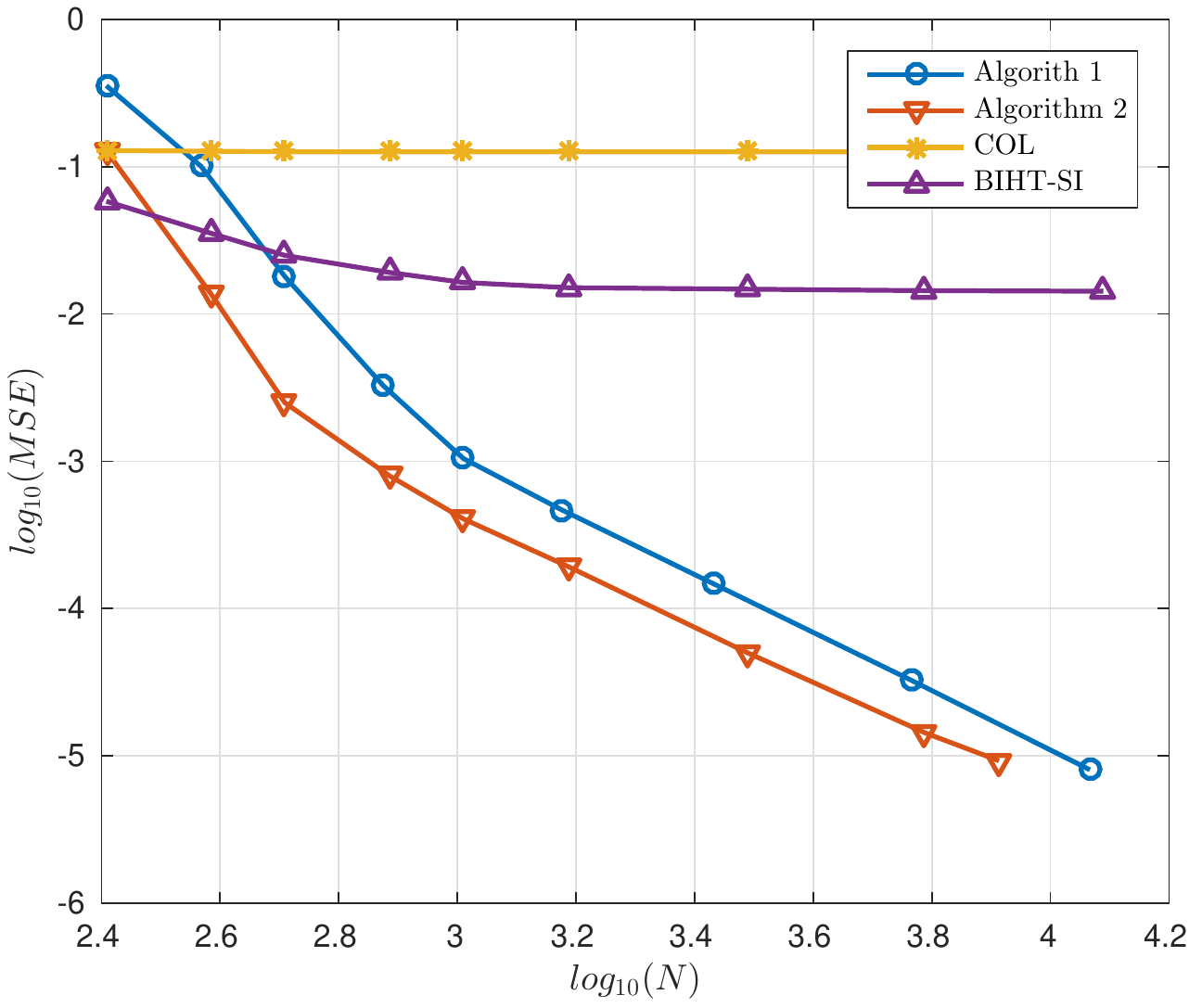} 
\caption{IEEE logo}
\label{fig:sub-second}
\end{subfigure}
\newline
\begin{subfigure}{.5\textwidth}
\centering
%
%
\includegraphics[height= 0.6\linewidth, width=.85\linewidth]{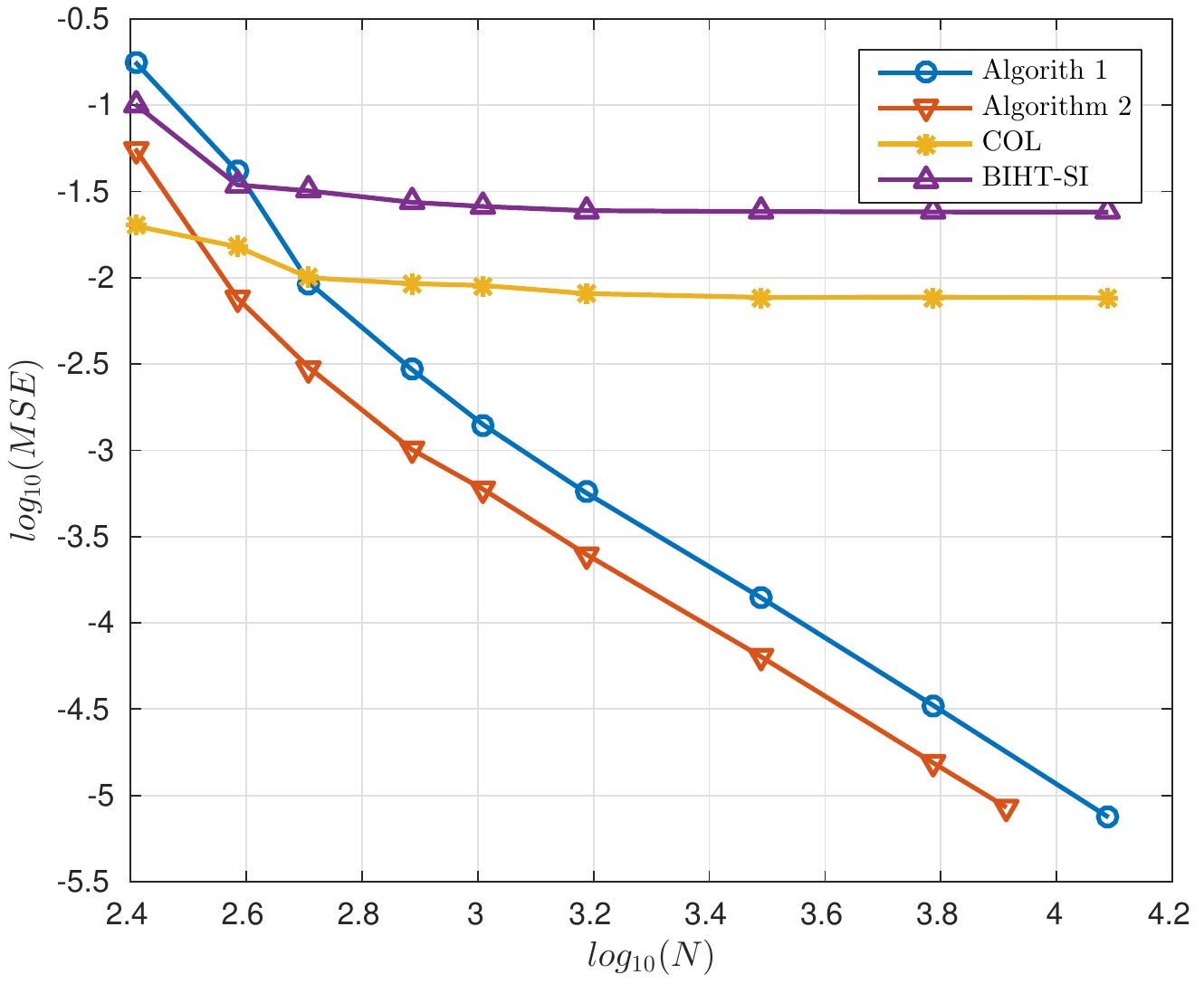} 
\caption{Cameraman}
\label{fig:sub-third}
\end{subfigure}
\begin{subfigure}{.5\textwidth}
\centering
%
%
\includegraphics[height= 0.6\linewidth, width=.85\linewidth]{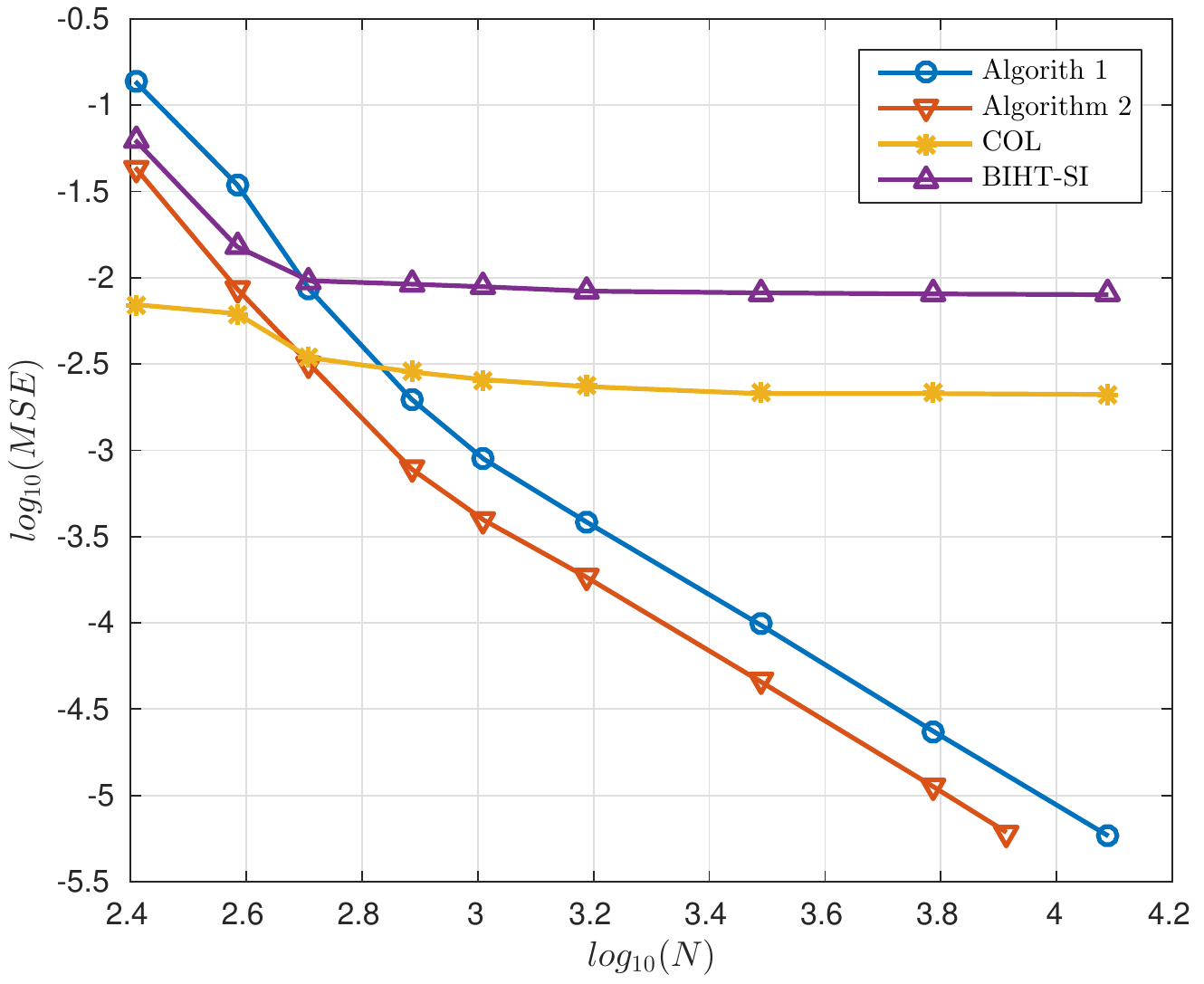} 
\caption{Peppers}
\label{fig:sub-fourth}
\end{subfigure}
\caption{Plots of $\log (MSE)$ v/s $\log (N)$. Here $M$ is fixed to $128$ and
$N^2/M^2$ is the oversampling ratio for the images. The asymptotic slopes of the
plots corresponding to Algorithm~\ref{algo:1} and~\ref{algo:2}  are given in
Table~\ref{table:slopes}.}
\label{fig:fig}
\end{figure*}
%
%
\begin{table}[hbt!]
%
%
\begin{tabular}{|P{0.09\textwidth}||P{0.06\textwidth}|P{0.06\textwidth}|P{0.07\textwidth}|P{0.08\textwidth}|
}
\hline
\multicolumn{5}{|c|}{Lena Image ($\sigma^2 = 100$)} \\
\hline
Method & PSNR & SSIM & ESSIM & MS-SSIM\\
\hline
Algorithm~\ref{algo:1}    &  37.083  &   0.972   &   0.941  &   0.986\\
Algorithm~\ref{algo:2}    &  39.142  &   0.978   &   0.943  &   0.988\\ 
COL &  21.800  &   0.696   &   0.893  &   0.882\\
BIHT-SI  &  18.426  &   0.584   &   0.878  &   0.786\\
\hline
\end{tabular}
\caption{Comparison of the algorithms on the Lena image on the basis of PSNR,
SSIM, ESSIM, and MS-SSIM.}
\label{table:lena}
\end{table}
%
%
\begin{table}[hbt!]
%
%
\begin{tabular}{|P{0.09\textwidth}||P{0.06\textwidth}|P{0.06\textwidth}|P{0.07\textwidth}|P{0.08\textwidth}|
}
\hline
\multicolumn{5}{|c|}{IEEE logo Image ($\sigma^2 = 100$)} \\
\hline
Method & PSNR & SSIM & ESSIM & MS-SSIM\\
\hline
Algorithm~\ref{algo:1}   &  37.430  &   0.989   &   0.868  &   0.975\\
Algorithm~\ref{algo:2}   &  37.190  &   0.991   &   0.902  &   0.984\\ 
COL &  22.202  &   0.825   &   0.775  &   0.923\\
BIHT-SI  &  17.880  &   0.706   &   0.734  &   0.821\\
\hline
\end{tabular} 
\caption{Comparison of the algorithms on the IEEE logo image on the basis of
PSNR, SSIM, ESSIM, and MS-SSIM.}
\label{table:logo}
\end{table}
%
%
\begin{table}[hbt!]
%
%
\begin{tabular}{|P{0.09\textwidth}||P{0.06\textwidth}|P{0.06\textwidth}|P{0.07\textwidth}|P{0.08\textwidth}|
}
 \hline
 \multicolumn{5}{|c|}{Cameraman Image ($\sigma^2 = 100$)} \\
 \hline
 Method & PSNR & SSIM & ESSIM & MS-SSIM\\
 \hline
 Algorithm~\ref{algo:1}    &  37.133  &   0.963   &   0.906  &   0.905\\
  Algorithm~\ref{algo:2}    &  37.511  &   0.969   &   0.905  &   0.916\\ 
 COL &  20.445  &   0.602   &   0.851  &   0.681\\
 BIHT-SI  &  16.293  &   0.450   &   0.796  &   0.613\\
\hline
\end{tabular}
\caption{Comparison of the algorithms on the cameraman image on the basis of
PSNR, SSIM, ESSIM, and MS-SSIM.}
\label{table:cm}
\end{table}
%
%
\begin{table}[hbt!]
%
%
\begin{tabular}{|P{0.09\textwidth}||P{0.06\textwidth}|P{0.06\textwidth}|P{0.07\textwidth}|P{0.08\textwidth}|
}
 \hline
 \multicolumn{5}{|c|}{Peppers Image ($\sigma^2 = 100$)} \\
 \hline
 Method & PSNR & SSIM & ESSIM & MS-SSIM\\
 \hline
 Algorithm~\ref{algo:1}    &  39.546  &   0.977   &   0.961  &   0.983\\
 Algorithm~\ref{algo:2}    &  39.302  &   0.978   &   0.965  &   0.986\\ 
 COL &  26.114  &   0.797   &   0.903  &   0.842\\
 BIHT-SI  &  19.830  &   0.592   &   0.867  &   0.701\\
 \hline
\end{tabular}
\caption{Comparison of the algorithms on the peppers image on the basis of PSNR,
SSIM, ESSIM, and MS-SSIM.}
 \label{table:pep}
\end{table}
%
%

\section{Conclusions}
\label{sec:conclusion}

In this paper, we propose two variations of a novel algorithm for the
reconstruction of signals using one-bit or two-bit noisy recordings of the
\ul{phase of its} 2D-DFT. The signal has zero-mean additive noise of a known
\revv{symmetric} distribution. We use Banach's contraction mapping theorem to
provide a recursion that converges to a close estimate of the signal. The
expected mean squared error in reconstruction is shown to be $O(M^2/N^2)$, where
$N^2/M^2$ is the oversampling ratio for the image. The result is validated via
numerical simulations on four benchmark images. Directions for future work includes developing signal reconstruction algorithms when the noise distributions are not known.

\appendix
\subsection{Proof of lemmas for Algorithm~\ref{algo:1}}
\label{appendix:A}
\begin{lemma} \label{lem:T}
The map $\T$ is a contraction on the set of real matrices \edit{in $S^{N \times N},$} with the Frobenius distance as the metric.
\end{lemma}
\begin{proof} 
Recall the map $\T$ as given in eq.~\eqref{eq:T}. We need to show that the Frobenius distance between two matrices \revv{$\g_1 \in \S^{N \times N}$ and $\g_2 \in \S^{N \times N}$} decreases on the application of $\T$,
\begin{align}
 \left \| \T(\g_1)-\T(\g_2) \right \|_F &= \left \|\P  \left( \g_1 - \g_2 -
\gamma\left( \F(\g_1)-\F(\g_2)\right) \right) \right\|_F. \nonumber
\end{align}
By the \textsc{Non-Expansive} property of $\P$,
\begin{align}
 \|  \T(\g_1)-\T(\g_2) \|_F &\leq \left \|  \g_1 - \g_2 - \gamma\left( \F(\g_1)-\F(\g_2)\right) \right\|_F. \nonumber
\end{align}
By the Lagrange mean value theorem, $\left(\F(\g_1)-\F(\g_2) \right) = f(c)(\g_1 - \g_2)$ for some $c$ such that each entry of $c$ is between the corresponding entries of $\g_1$ and $\g_2$. Thus, we get,
\begin{align}
 \| \T(\g_1)- \T(\g_2) \|_F & \leq \| 1-\gamma f\| _{\max} \left\| (\g_1 - \g_2)\right\|_F.\nonumber
\end{align}
\ul{where $f(x)$ is the derivative of $\F(x)$, i.e., $f(x)$ is the probability
density function of the noise.} Recall the definition $\alpha := \|1-\gamma f\|_{\max}$. For $\T$ to be a contraction, we require $ 0 < \alpha < 1$. This is
ensured by restricting $\gamma$ to $\left(0,\frac{2}{f_{\mbox{\footnotesize max}}}\right)$. 
\ul{Recall that $f_{\mbox{\footnotesize max}}$ is the maximum value of $f(x)$ in $x \in (-\infty, \infty)$}
\qedhere
\end{proof}
\begin{lemma} \label{lem:var1}
The average variance of $\P(\X_R)$ is $O(M^2/N^2)$.
\end{lemma}
\begin{proof}
\st{Recall that the function $\P$ projects the frequency domain argument to the
support region in the spatial domain.}
Recall that the vectorized form of matrix $g[n_1,n_2]$ is given by $g^v[n]$ and
the vectorization operation on a matrix corresponds to concatenating its columns
in order.  We know that the 2D-DFT is an orthogonal
transform~\cite{jain1989fundamentals}. Therefore, the operation on $g^v[n]$
equivalent to 2D-DFT of $g[n_1,n_2]$ can be given by
\begin{align}
\tilde{g}^v = N \U g^v. \label{def:U}
\end{align}
Here vector $\tilde{g}^v[k]$ is the vectorized form of $\g[k_1,k_2]$. Matrix $\U
\in \mathbb{C}^{N^2 \times N^2}$ is unitary and has orthonormal columns. Denote
the columns of $\U$ by $u_1, u_2,\ldots,u_{N^2}$. Similarly, the 2D-IDFT of
$\g[k_1,k_2]$ can be expressed as an orthogonal transform of $\g^v[k]$, given in the terms of the complex conjugate of $\U$ as:

\begin{align}
g^v = \frac{1}{N}\U^*\tilde{g}^v.  \label{def:Ustar}
\end{align} 
Let the operation on $\g^v[k]$, equivalent to the projection operation $\P$ on
$\g[k_1,k_2]$, be given by $\P^v(\g^v)$. In this paper, for an image of size $M
\times M$ pixels, a $N\times N$ size 2D-DFT is computed. Here $N>2M$. The
2D-IDFT of the real-part of the frequency domain image gives two copies of the
reconstructed image in the spatial domain. Thus, there are $2M^2$ entries in the
support region in the spatial domain. After vectorization, let the indices
corresponding to these support region entries be in the set $\mathcal{M}$. Thus
the cardinality of $\mathcal{M}$ is $2M^2$.

By definition $\X_R^v[i] \in \left\{-\frac{1}{2}, \frac{1}{2}\right\}$.  For
each $i \in \{1,2,\ldots,N^2\}$, $\mbox{var} (\X_R^v[i]) \leq \frac{1}{4}$. Let
$S^v = \P^v(\X_R^v)$ be the vectorized form of $\P(\X_R)$. Recall that $\P$
consists of computing the 2D-IDFT, followed by projection onto the support set
in the spatial domain, \edit{clipping to the signal range in the spatial domain,} and then computing the 2D-DFT. Its vectorized version
$\P^v$ can be given as\edit{\footnote{\edit{Recall the proof of Lemma~\ref{lem:exp-S}, in which we showed that $\P(\X_R) = \dft{\textsc{Proj}(\idft{\X_R})}$ since \textsc{Clip} does not change $\textsc{Proj}(\idft{\X_R}). $ This ensures that $\P(\X_R)$ is a linear transform of $\X_R,$ although $\P$ is not a linear operator in general because of the $\textsc{Clip}$ operation within it.
}}}: 
\begin{align}
& S^v[k] =\P^v(\X_R^v)=  \sum_{n=1}^{N^2} \left( \sum_{i \in \mathcal{M}}
\u_i^*[k] \u_n[i]\right) \X_R^v[n],\nonumber
\intertext{Since $W_R, d_R$ are i.i.d., the elements of $\X_R^v$ are
independent,}
&\mbox{var}(S^v[k])  \leq \sum_{n=1}^{N^2} \left( \sum_{i \in \mathcal{M}}
\u_i^*[k] \u_n[i]\right)^2 \mbox{var} (\X_R^r[n]), \nonumber \\
 &\leq \sum_{n=1}^{N^2} \frac{1}{4}\left( \sum_{i \in \mathcal{M}} \u_i^*[k]
\u_n[i]\right)^2, \nonumber \\
&= \frac{1}{4} \sum_{n=1}^{N^2} \sum_{i \in \mathcal{M}} \sum_{j \in
\mathcal{M}} \u_i^*[k]\u_j^*[k]\u_n[j]\u_n[i]. \nonumber\\
\intertext{For unitary matrices, $\U^\intercal = \U^* $ and thus, $\u_n[j] =
u_j^*[n]$,}
&\mbox{var}(S^v[k]) \leq \frac{1}{4} \sum_{n=1}^{N^2} \sum_{i \in \mathcal{M}}
\sum_{j \in \mathcal{M}} \u_i^*[k]\u_j^*[k]u_j^*[n]u_i^*[n]. \nonumber \\
\intertext{By changing the order of summations, we get,} \nonumber
%
&\leq \frac{1}{4} \sum_{i \in \mathcal{M}} \sum_{j \in
\mathcal{M}} \u_i^*[k]\u_j^*[k] \delta[i-j] = \frac{1}{4} \sum_{i \in
\mathcal{M}} \left( \u_i^*[k]\right)^2. \nonumber\\
%
%
\intertext{To find the mean variance, we average over $k$ on both sides,}
%
%
%
&\frac{1}{N^2}\sum_{k=1}^{N^2} \mbox{var}(S^v[k]) \leq \frac{1}{4N^2}\sum_{i \in
\mathcal{M}}  \sum_{k=1}^{N^2} \left( \u_i^*[k]\right)^2. \nonumber \\
\intertext{By the unit norm of the columns of orthonormal matrix $\U^*$ and
Since $\mathcal{M}$ has $2M^2$ entries,} 
& \frac{1}{N^2}\sum_{k=1}^{N^2} \mbox{var}(S^v[k])  \leq \frac{1}{4N^2}\sum_{i
\in \mathcal{M}} 1 = \frac{M^2}{2N^2}. \nonumber
%
%
\end{align} 
%
%
This completes the proof.
%
\end{proof}

\subsection{Proof of lemmas for Algorithm~\ref{algo:2}}
\label{appendix:B}
\begin{lemma} \label{lem:Z}
The map $\Z$ is a contraction on the set of complex matrices \edit{in $\S_c^{N \times N}$} with the Frobenius distance as the metric.
\end{lemma}
\begin{proof} Recall the map $\Z$ as given in eq.~\eqref{eq:Z}. We need to show
that the Frobenius distance between two complex matrices $\g^{(1)}, \g^{(2)} \in \S_c^{N \times N}$ decreases under the
application of $\Z$.  \begin{align}
 &\left \| \Z(\g^{(1)})- \Z(\g^{(2)}) \right \|_F =\Big \| \Q  \left( \g^{(1)} -
\g^{(2)} \right)  \nonumber\\
 &  - \gamma \Q\left( \F(\g_R^{(1)})-\F(\g_R^{(2)})  + j\F(\g_I^{(1)}) -
j\F(\g_I^{(2)})  \right) \Big\|_F.\nonumber 
\end{align}
By the \textsc{Non-Expansive} property of $\Q$,
\begin{align}
 & \left \| \Z(\g^{(1)})- \Z(\g^{(2)}) \right \|_F  \leq  \Big \|  \left(
\g^{(1)} - \g^{(2)}  \right)\nonumber \\
 &  - \gamma\left( \F(\g_R^{(1)})-\F(\g_R^{(2)})  + j\F(\g_I^{(1)}) -
j\F(\g_I^{(2)}) \right) \Big\|_F, \nonumber
 \end{align}
By the Lagrange mean value theorem, as used in Lemma~\ref{lem:T}, 
 \begin{align}
%
 &\leq \|1-\gamma f\|_{\max}  \left \|( \g_R^{(1)} - \g_R^{(2)} ) + j(
\g_I^{(1)} - \g_I^{(2)})\right \|_F, \nonumber\\
 &\leq \|1-\gamma f\|_{\max} \left \| \g^{(1)} - \g^{(2)}\right \|_F.\nonumber
\end{align}
Recall the definition $\alpha = \|1-\gamma f\|_{\max}$. For $\Z$ to be a
contraction map, we require $ 0 < \alpha < 1$. This is ensured by restricting
$\gamma$ to $\left(0,\frac{2}{f_{\mbox{\footnotesize max}}}\right)$. Recall that
$f_{\mbox{\footnotesize max}}$ is the maximum value of $f(x)$ in $x \in
(-\infty, \infty)$. 
\end{proof}
\begin{lemma} \label{lem:var2}
The average variance of $\Q(\X)$ is $O(M^2/N^2)$.
\end{lemma}

\begin{proof}
The proof is similar to that of Lemma~\ref{lem:var1}. Recall that the operator
$\Q$ projects the argument to its support region in the spatial domain. For a
frequency domain argument, it first computes its 2D-IDFT, then sets the pixels
outside the support region to zero, and then finally computes its 2D-DFT.  For a
2D-DFT of size $N \times N$ pixels, of an image of size $M \times M$ pixels,
there will be $N^2-M^2$ entries corresponding to the zero-padding, i.e., outside
the support region in the spatial domain. Recall that the vectorized form of
matrix $g[n_1,n_2]$ is given by $g^v[n]$ and the vectorization operation on a
matrix corresponds to concatenating its columns in order. Since the 2D-DFT is an
orthogonal transform~\cite{jain1989fundamentals}, the operation on $g^v[n]$
equivalent to 2D-DFT of $g[n_1,n_2]$ can be given as in eq.~\eqref{def:U}.
Similarly, the 2D-IDFT of $\g[k_1,k_2]$ can be expressed as an orthogonal
transform of $\g^v[k]$, as in eq.\eqref{def:Ustar}.

Let the operation on $\g^v[k]$, equivalent to the projection operation $\Q$ on
$\g[k_1,k_2]$, be given by $\Q^v(\g^v)$. After vectorization, let the indices
corresponding to the support region in the spatial domain be in the set
$\mathcal{M}$. Here the cardinality of $\mathcal{M}$ is $M^2$.
 
By definition $\X_R^v[i] \in \{-\frac{1}{2}, \frac{1}{2}\}$ and $\X_I^v[i] \in
\{-\frac{1}{2}, \frac{1}{2}\}$.  For each $i \in \{1,2,\ldots,N^2\}$,
$\mbox{var} (\X_R^v[i]) \leq \frac{1}{4}$ and $\mbox{var} (\X_I^v[i]) \leq
\frac{1}{4}$. Let $\Y^v = \Q^v(\X)$. Since $\Q$ consists of computing 2D-DFT,
followed by projection onto the support set and then computing 2D-IDFT, its
vectorized version can be given as\edit{\footnote{\edit{Recall the proof of Lemma~\ref{lem:exp-Y}, in which we showed that $\Q(\X) = \dft{\textsc{Proj}(\idft{\X})}$ since \textsc{Clip} does not change $\textsc{Proj}(\idft{\X}).$ This ensures that $\Q(\X)$ is a linear transform of $\X,$ although $\Q$ is not a linear operator in general because of the $\textsc{Clip}$ operation which is a part of $\Q.$
}}}:
\begin{align}
\Y^v[k] =  \sum_{n=1}^{N^2} \left( \sum_{i \in \mathcal{M}} \u_i^*[k]
\u_n[i]\right) \left(\X_R^v[n]+ j \X_I^v[n] \right).\nonumber
\end{align}
Here $u_i, ~\forall ~i \in [1, \cdots, N^2],$ are the orthonormal columns of $\U$. Since $W_d, d_R, W_I,$ and $d_I$ are i.i.d., the elements of $\X_R^v$ and $\X_I^v$ are independent. Therefore,
\begin{align}
&\mbox{var}(\Y^v[k]) \nonumber\\
& \leq \sum_{n=1}^{N^2} \left( \sum_{i \in \mathcal{M}} \u_i^*[k] \u_n[i]\right)^2 \left( \mbox{var} (\X_R^v[n]) + \mbox{var} (\X_i^v[n])  \right),\nonumber  \\
 &= \sum_{n=1}^{N^2} \left( \frac{1}{4}+ \frac{1}{4} \right)\left( \sum_{i \in \mathcal{M}} \u_i^*[k] \u_n[i]\right)^2,  \nonumber\\
&= \frac{1}{2} \sum_{n=1}^{N^2} \sum_{i \in \mathcal{M}} \sum_{j \in \mathcal{M}}
\u_i^*[k]\u_j^*[k]\u_n[j]\u_n[i]. \nonumber\\
\intertext{Since $\U$ is a unitary matrices, $\U^\intercal = \U^*, \u_n[j] = u_j^*[n],$ then} 
&\mbox{var}(\Y^v[k]) \leq \frac{1}{2} \sum_{n=1}^{N^2} \sum_{i \in \mathcal{M}} \sum_{j \in \mathcal{M}}
\u_i^*[k]\u_j^*[k]u_j^*[n]u_i^*[n].  \nonumber\\
\intertext{Changing the order of summations,} 
%
&\leq \frac{1}{2} \sum_{i \in \mathcal{M}} \sum_{j \in \mathcal{M}} \u_i^*[k]~\u_j^*[k] ~\delta[i-j] = \frac{1}{2} \sum_{i \in \mathcal{M}} \left( \u_i^*[k]\right)^2. \nonumber\\
%
%
\intertext{Being interested in the mean variance, we average over $k$,}
%
%
%
&\frac{1}{N^2}\sum_{k=1}^{N^2} \mbox{var}(\Y^v[k]) \leq \frac{1}{2N^2}\sum_{i \in \mathcal{M}}  \sum_{k=1}^{N^2} \left( \u_i^*[k]\right)^2. \nonumber \\
\intertext{By the unit norm of columns of orthonormal matrix $\U^*$ and since $\mathcal{M}$ has $M^2$ entries,} \nonumber 
&\frac{1}{N^2}\sum_{k=1}^{N^2} \mbox{var}(\Y^v[k]) \leq \frac{1}{2N^2}\sum_{i \in \mathcal{M}} 1 = \frac{M^2}{2N^2}. \nonumber
%
%
%
\end{align}
This completes the proof. \qedhere
\end{proof}

\bibliographystyle{IEEEtran}
\bibliography{references}

\begin{thebibliography}{10}
\providecommand{\url}[1]{#1}
\csname url@samestyle\endcsname
\providecommand{\newblock}{\relax}
\providecommand{\bibinfo}[2]{#2}
\providecommand{\BIBentrySTDinterwordspacing}{\spaceskip=0pt\relax}
\providecommand{\BIBentryALTinterwordstretchfactor}{4}
\providecommand{\BIBentryALTinterwordspacing}{\spaceskip=\fontdimen2\font plus
\BIBentryALTinterwordstretchfactor\fontdimen3\font minus
  \fontdimen4\font\relax}
\providecommand{\BIBforeignlanguage}[2]{{%
\expandafter\ifx\csname l@#1\endcsname\relax
\typeout{** WARNING: IEEEtran.bst: No hyphenation pattern has been}%
\typeout{** loaded for the language `#1'. Using the pattern for}%
\typeout{** the default language instead.}%
\else
\language=\csname l@#1\endcsname
\fi
#2}}
\providecommand{\BIBdecl}{\relax}
\BIBdecl

\bibitem{oppenheim1983signal}
A.~V. Oppenheim, J.~S. Lim, and S.~R. Curtis, ``Signal synthesis and
  reconstruction from partial {F}ourier-domain information,'' \emph{JOSA},
  vol.~73, no.~11, pp. 1413--1420, 1983.

\bibitem{fienup2013phase}
J.~R. Fienup, ``Phase retrieval algorithms: a personal tour,'' \emph{Applied
  optics}, vol.~52, no.~1, pp. 45--56, 2013.

\bibitem{saxton2013computer}
W.~Saxton, \emph{Computer techniques for image processing in electron
  microscopy}.\hskip 1em plus 0.5em minus 0.4em\relax Academic Press, 2013,
  vol.~10.

\bibitem{ramachandran1970fourier}
G.~N. Ramachandran, G.~N. Ramachandran, and R.~Srinivasan, \emph{{F}ourier
  methods in crystallography}.\hskip 1em plus 0.5em minus 0.4em\relax John
  Wiley \& Sons, 1970.

\bibitem{candes2013phaselift}
E.~J. Candes, T.~Strohmer, and V.~Voroninski, ``Phaselift: Exact and stable
  signal recovery from magnitude measurements via convex programming,''
  \emph{Communications on Pure and Applied Mathematics}, vol.~66, no.~8, pp.
  1241--1274, 2013.

\bibitem{huang2016phase}
K.~Huang, Y.~C. Eldar, and N.~D. Sidiropoulos, ``Phase retrieval from 1d
  {F}ourier measurements: Convexity, uniqueness, and algorithms,'' \emph{IEEE
  Transactions on Signal Processing}, vol.~64, no.~23, pp. 6105--6117, 2016.

\bibitem{kishore2020wirtinger}
V.~{Kishore} and C.~S. {Seelamantula}, ``Wirtinger flow algorithms for phase
  retrieval from binary measurements,'' in \emph{IEEE International Conference
  on Acoustics, Speech and Signal Processing}, 2020.

\bibitem{mukherjee2018phase}
S.~Mukherjee and C.~S. Seelamantula, ``Phase retrieval from binary
  measurements,'' \emph{IEEE Signal Processing Letters}, 2018.

\bibitem{li1983arrival}
Y.~Li and A.~Kurkjian, ``Arrival time determination using iterative signal
  reconstruction from the phase of the cross spectrum,'' \emph{IEEE
  Transactions on Acoustics, Speech, and Signal Processing}, vol.~31, no.~2,
  pp. 502--504, 1983.

\bibitem{curtis1985signal}
S.~Curtis, A.~Oppenheim, and J.~Lim, ``Signal reconstruction from {F}ourier
  transform sign information,'' \emph{IEEE transactions on acoustics, speech,
  and signal processing}, vol.~33, no.~3, pp. 643--657, 1985.

\bibitem{tang1990image}
X.~Tang, Y.~Yuan, and Y.~Wang, ``Image reconstruction from one-bit phase
  information (obpi) with specified histogram constraint,'' in \emph{IEEE
  International Symposium on Circuits and Systems}, 1990, pp. 755--758.

\bibitem{lyuboshenko1999stable}
I.~Lyuboshenko and A.~Akhmetshin, ``Stable signal and image reconstruction from
  noisy {F}ourier transform phase,'' \emph{IEEE transactions on signal
  processing}, vol.~47, no.~1, pp. 244--250, 1999.

\bibitem{thomas1984procedures}
D.~Thomas and M.~Hayes, ``Procedures for signal reconstruction from noisy
  phase,'' in \emph{IEEE International Conference on Acoustics, Speech, and
  Signal Processing}, vol.~9, 1984, pp. 618--621.

\bibitem{gray1987oversampled}
R.~Gray, ``Oversampled sigma-delta modulation,'' \emph{IEEE Transactions on
  Communications}, vol.~35, no.~5, pp. 481--489, 1987.

\bibitem{masry}
E.~Masry, ``The reconstruction of analog signals from the sign of their noisy
  samples,'' \emph{IEEE Transactions on Information Theory}, vol.~27, no.~6,
  pp. 735--745, 1981.

\bibitem{thong2002nonlinear}
T.~Thong and J.~McNames, ``Nonlinear reconstruction of over-sampled coarsely
  quantized signals,'' in \emph{The 45th Midwest Symposium on Circuits and
  Systems}, vol.~2.\hskip 1em plus 0.5em minus 0.4em\relax IEEE, 2002.

\bibitem{daubechies2003approximating}
I.~Daubechies and R.~DeVore, ``Approximating a bandlimited function using very
  coarsely quantized data: A family of stable sigma-delta modulators of
  arbitrary order,'' \emph{Annals of mathematics}, vol. 158, no.~2, pp.
  679--710, 2003.

\bibitem{kumar2013estimation}
A.~Kumar and V.~M. Prabhakaran, ``Estimation of bandlimited signals from the
  signs of noisy samples,'' in \emph{IEEE International Conference on
  Acoustics, Speech and Signal Processing,}, 2013, pp. 5815--5819.

\bibitem{cvetkovic2000single}
Z.~Cvetkovic and I.~Daubechies, ``Single-bit oversampled a/d conversion with
  exponential accuracy in the bit-rate,'' in \emph{Data Compression Conference,
  2000. Proceedings. DCC 2000}.\hskip 1em plus 0.5em minus 0.4em\relax IEEE,
  2000, pp. 343--352.

\bibitem{khobahi2019deep}
S.~{Khobahi}, N.~{Naimipour}, M.~{Soltanalian}, and Y.~C. {Eldar}, ``Deep
  signal recovery with one-bit quantization,'' in \emph{2019 IEEE International
  Conference on Acoustics, Speech and Signal Processing}, 2019.

\bibitem{bender2020spectral}
S.~{Bender}, M.~{Dörpinghaus}, and G.~{Fettweis}, ``On the spectral efficiency
  of bandlimited 1-bit quantized awgn channels with runlength-coding,''
  \emph{IEEE Communications Letters}, 2020.

\bibitem{shao2019channel}
Z.~{Shao}, L.~T.~N. {Landau}, and R.~C. {de Lamare}, ``Channel estimation using
  1-bit quantization and oversampling for large-scale multiple-antenna
  systems,'' in \emph{IEEE International Conference on Acoustics, Speech and
  Signal Processing (ICASSP)}, 2019, pp. 4669--4673.

\bibitem{goyal2018estimation}
M.~Goyal and A.~Kumar, ``Estimation of bandlimited signals on graphs from
  single bit recordings of noisy samples,'' in \emph{2018 26th European Signal
  Processing Conference (EUSIPCO)}.\hskip 1em plus 0.5em minus 0.4em\relax
  IEEE, 2018, pp. 902--906.

\bibitem{boufounos20081}
P.~T. Boufounos and R.~G. Baraniuk, ``1-bit compressive sensing,'' in
  \emph{IEEE Conference on Information Sciences and Systems}, 2008.

\bibitem{zymnis2010compressed}
A.~Zymnis, S.~Boyd, and E.~Candes, ``Compressed sensing with quantized
  measurements,'' \emph{IEEE Signal Processing Letters}, vol.~17, no.~2, pp.
  149--152, 2010.

\bibitem{xu2018quantized}
C.~Xu and L.~Jacques, ``Quantized compressive sensing with rip matrices: The
  benefit of dithering,'' \emph{Information and Inference: A Journal of the
  IMA}, 2018.

\bibitem{jacques2013robust}
L.~Jacques, J.~N. Laska, P.~T. Boufounos, and R.~G. Baraniuk, ``Robust 1-bit
  compressive sensing via binary stable embeddings of sparse vectors,''
  \emph{IEEE Transactions on Information Theory}, vol.~59, no.~4, pp.
  2082--2102, 2013.

\bibitem{friedlander2020nbiht}
M.~P. Friedlander, H.~Jeong, Y.~Plan, and O.~Yilmaz, ``{NBIHT:} an efficient
  algorithm for 1-bit compressed sensing with optimal error decay rate,''
  \emph{arXiv preprint arXiv:2012.12886}, 2020.

\bibitem{boufounos2013angle}
P.~T. Boufounos, ``Angle-preserving quantized phase embeddings,'' in
  \emph{Wavelets and Sparsity XV}, vol. 8858.\hskip 1em plus 0.5em minus
  0.4em\relax International Society for Optics and Photonics, 2013.

\bibitem{boufounos2013sparse}
------, ``Sparse signal reconstruction from phase-only measurements,'' in
  \emph{Proc. Int. Conf. Sampling Theory and Applications}.\hskip 1em plus
  0.5em minus 0.4em\relax Citeseer, 2013.

\bibitem{jacques2020importance}
L.~Jacques and T.~Feuillen, ``The importance of phase in complex compressive
  sensing,'' \emph{arXiv preprint arXiv:2001.02529}, 2020.

\bibitem{seber2012linear}
G.~A. Seber and A.~J. Lee, \emph{Linear regression analysis}.\hskip 1em plus
  0.5em minus 0.4em\relax John Wiley \& Sons, 2012, vol. 329.

\bibitem{goyal1998quantized}
V.~K. Goyal, M.~Vetterli, and N.~T. Thao, ``Quantized overcomplete expansions
  in ir/sup n: analysis, synthesis, and algorithms,'' \emph{IEEE Transactions
  on Information Theory}, vol.~44, no.~1, pp. 16--31, 1998.

\bibitem{ai2014one}
A.~Ai, A.~Lapanowski, Y.~Plan, and R.~Vershynin, ``One-bit compressed sensing
  with non-gaussian measurements,'' \emph{Linear Algebra and its Applications},
  vol. 441, pp. 222--239, 2014.

\bibitem{plan2013one}
Y.~Plan and R.~Vershynin, ``One-bit compressed sensing by linear programming,''
  \emph{Communications on Pure and Applied Mathematics}, vol.~66, no.~8, pp.
  1275--1297, 2013.

\bibitem{so2018reconstruction}
S.~So and K.~K. Paliwal, ``Reconstruction of a signal from the real part of its
  discrete {F}ourier transform [tips \& tricks],'' \emph{IEEE Signal Processing
  Magazine}, vol.~35, no.~2, pp. 162--174, 2018.

\bibitem{bahmani2013robust}
S.~Bahmani, P.~T. Boufounos, and B.~Raj, ``Robust 1-bit compressive sensing via
  gradient support pursuit,'' \emph{arXiv preprint arXiv:1304.6627}, 2013.

\bibitem{wagdy1989effect}
M.~F. Wagdy, ``Effect of various dither forms on quantization errors of ideal
  a/d converters,'' \emph{IEEE Transactions on Instrumentation and
  Measurement}, vol.~38, no.~4, pp. 850--855, 1989.

\bibitem{kreyszig1989introductory}
E.~Kreyszig, \emph{Introductory functional analysis with applications}.\hskip
  1em plus 0.5em minus 0.4em\relax wiley New York, 1989, vol.~1.

\bibitem{north2015one}
P.~North and D.~Needell, ``One-bit compressive sensing with partial support,''
  in \emph{2015 IEEE 6th International Workshop on Computational Advances in
  Multi-Sensor Adaptive Processing}, 2015, pp. 349--352.

\bibitem{guan2006edge}
G.-H. Chen, C.-L. Yang, L.-M. Po, and S.-L. Xie, ``Edge-based structural
  similarity for image quality assessment,'' in \emph{IEEE International
  Conference on Acoustics Speech and Signal Processing Proceedings}, 2006.

\bibitem{wang2004ssim}
Z.~Wang, A.~C. Bovik, H.~R. Sheikh, and E.~P. Simoncelli, ``Image quality
  assessment: from error visibility to structural similarity,'' \emph{IEEE
  Transactions on Image Processing}, vol.~13, no.~4, pp. 600--612, 2004.

\bibitem{wang2003multi}
Z.~Wang, E.~P. Simoncelli, and A.~C. Bovik, ``Multiscale structural similarity
  for image quality assessment,'' in \emph{The Asilomar Conference on Signals,
  Systems Computers}, vol.~2, 2003, pp. 1398--1402 Vol.2.

\bibitem{jain1989fundamentals}
A.~K. Jain, \emph{Fundamentals of digital image processing}.\hskip 1em plus
  0.5em minus 0.4em\relax Englewood Cliffs, NJ: Prentice Hall,, 1989.

\end{thebibliography}

\end{document}